\documentclass{article}

\usepackage{arxiv}
\usepackage{amsthm}
\usepackage{amsmath,amssymb}
\usepackage{xcolor}
\usepackage{comment}

\usepackage[utf8]{inputenc} % allow utf-8 input
\usepackage[T1]{fontenc}    % use 8-bit T1 fonts
\usepackage{hyperref}       % hyperlinks
\usepackage{cleveref}
\usepackage{url}            % simple URL typesetting
\usepackage{booktabs}       % professional-quality tables
\usepackage{amsfonts}       % blackboard math symbols
\usepackage{nicefrac}       % compact symbols for 1/2, etc.
\usepackage{microtype}      % microtypography
\usepackage{lipsum}		% Can be removed after putting your text content
\usepackage{graphicx}
\usepackage{doi}
\usepackage{amsmath}
\usepackage{cleveref}

\allowdisplaybreaks

\newcommand{\alphabet}[1] { {\mathsf #1}}

\newcommand{\thetasijc}{\theta^*_{ij,c}}
\newcommand{\thetasijs}{\theta^*_{ij,s}}
\newcommand{\thetaijc}{\theta_{ij,c}}
\newcommand{\thetaijs}{\theta_{ij,s}}
\newcommand{\Xkijc}{X^{(k)}_{ij,c}}
\newcommand{\Xkijs}{X^{(k)}_{ij,s}}
\newcommand{\Xijc}{X_{ij,c}}
\newcommand{\Xijs}{X_{ij,s}}

\newcommand{\yk}{y^{(k)}}
\newcommand{\Yk}{Y^{(k)}}

\newcommand{\benum}{\begin{enumerate}}
\newcommand{\eenum}{\end{enumerate}}

\newcommand{\reals}{\mathbb{R}}
\newcommand{\parenth}[1] {\left(#1\right)}

\newcommand{\braces}[1] {\left\{#1\right\}}
\newcommand{\brackets}[1] {\left[#1\right]}
\newcommand{\abs}[1] {\left|#1\right|}
\newcommand{\norm}[1] {\left\|{#1}\right\|}

\newcommand{\cY}{\alphabet{Y}}
\newcommand{\vY}{\underline{Y}}
\newcommand{\vy}{\underline{y}}

\newcommand{\beqa}{\begin{eqnarray}}
\newcommand{\eeqa}{\end{eqnarray}}
\newcommand{\beqas}{\begin{eqnarray*}}
\newcommand{\eeqas}{\end{eqnarray*}}
\newcommand{\prob}[1] {\mathbb{P}\parenth{#1}}

\newcommand{\utheta}{\underline{\theta}}
\newcommand{\uthetau}{\utheta_u}

\newcommand{\ISO}{S_n(\uthetau)}

\newcommand{\sumijEu}{\sum_{(i,j) \in E_u}}

\newcommand{\ut}{\underline{t}}
\newcommand{\utk}{\ut^{(k)}}

\newcommand{\uphi}{\underline{\phi}}

\newcommand{\E}{\mathbb{E}}
\newcommand{\innerprod}[2]{\langle #1,#2 \rangle}
\newcommand{\given}[1]{\;#1|\;}

\newtheorem{theorem}{Theorem}[section]
\newtheorem{lemma}{Lemma}[section]
\newtheorem{condition}{Condition}[section]
\newtheorem{proposition}{Proposition}[section]
\newtheorem{corollary}{Corollary}[lemma]

\theoremstyle{definition}
\newtheorem{definition}{Definition}[section]

\title{Graphical models and efficient Inference methods for multivariate phase probability distributions}

\author{ \href{https://orcid.org/0000-0003-2962-6245}{\includegraphics[scale=0.06]{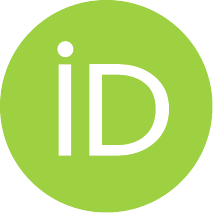}\hspace{1mm}Andrew S. Perley}\thanks{Use footnote for providing further
		information about author (webpage, alternative
		address)---\emph{not} for acknowledging funding agencies.} \\
	Department of Bioengineering\\
	Stanford University\\
	Stanford, CA 94305 \\
	\texttt{aperley@stanford.edu} \\
	\And
	\href{https://orcid.org/0000-0002-2144-2495}{\includegraphics[scale=0.06]{orcid.pdf}\hspace{1mm}Todd P. Coleman} \\
	Department of Bioengineering\\
	Stanford University\\
	Stanford, CA 94305 \\
	\texttt{toddcol@stanford.edu} 
}

\renewcommand{\shorttitle}{Graphical Models for Multivariate Phase Relationships}

\hypersetup{
pdftitle={A template for the arxiv style},
pdfsubject={q-bio.NC, q-bio.QM},
pdfauthor={Andrew S. Perley, Todd P. Coleman},
pdfkeywords={First keyword, Second keyword, More},
}

\begin{document}
\maketitle

\begin{abstract}
Multivariate phase relationships are important to characterize and understand  numerous physical, biological, and chemical systems, from electromagnetic waves to neural oscillations.  These systems exhibit complex spatiotemporal dynamics and intricate interdependencies among their constituent elements. While classical models of multivariate phase relationships, such as the wave equation and Kuramoto model, give theoretical models to describe phenomena, the development of statistical tools for hypothesis testing and inference for multivariate phase relationships in complex systems remains limited. This paper introduces a novel probabilistic modeling framework to characterize multivariate phase relationships, with wave-like phenomena serving as a key example.  This approach describes spatial patterns and interactions between oscillators through a pairwise exponential family distribution. Building upon the literature of graphical model inference, including methods like Ising models, graphical lasso, and interaction screening, this work bridges the gap between classical wave dynamics and modern statistical approaches. Efficient inference methods are introduced, leveraging the Chow-Liu algorithm for directed tree approximations and interaction screening for general graphical models. Simulated experiments demonstrate the utility of these methods for uncovering wave properties and sparse interaction structures, highlighting their applicability to diverse scientific domains. This framework establishes a new paradigm for statistical modeling of multivariate phase relationships, providing a powerful toolset for exploring the complexity of these systems.
\end{abstract}

% keywords can be removed

\section{Introduction}
Structured multivariate phase relationships are commonly observed in the natural world and are a key subject of interest across disciplines such as physics, chemistry, and biology.  One key example of structured phase relationships is that of waves, which underpin  countless natural and engineered systems from the sound and light to the oscillatory dynamics of neural processes and cardiac tissue \cite{surface_em_waves,plant_sound,cardiac_Waves,Muller_Chavane_Reynolds_Sejnowski_2018}. Classic equations like the wave equation, the Korteweg–de Vries (KdV) equation, and the Kuramoto model for coupled oscillators provide foundational insights into modeling multivariate phase relationships. However, while these models excel at describing these complex behavior theoretically, they do not provide statistical frameworks for analyzing empirical data, limiting their applicability to hypothesis testing and inference. Recent work in neuroscience, inspired by traveling waves in the brain, has made headway into reducing falste positives in detecting waves by physics-inspired means of surrogate testing methods to detect plane waves, rotating waves, and source/sink waves \cite{muller2016rotating}. Other methods ask the question if two different multivariate phase relationships are similar to each other by assessing their phase similarity and generating motifs \cite{davis2024horizontal}. While these methods are very powerful in reducing false positives in detecting specific pre-defined types of wave-like phenomena, they lack the flexibility of statistically inferring the multivariate phase relationships themselves by means of a probability distribution. 

The ability to probabilistically characterize multivariate phase relationships and uncover interactions among oscillating elements has profound implications. For example, identifying spatiotemporal patterns in neuronal oscillations can shed light on how spontaneous neural processes inter-relate in the resting state as well as how traveling waves mediate information transfer during motor tasks and gate sensory perception \cite{vidaurre2018spontaneous,rubino2006propagating,davis2020spontaneous,gonzales2025touch}.  As another example, characterizing multivariate phase statistics and their relationship with other covariates (e.g. neural firing) has important implications in describing how long-range neural interactions coordinate cell assemblies \cite{canolty2010oscillatory}.  
Additionally, wave patterns in gastric myoelectric activity have been previously shown to correlate with clinically validated gastric symptom severity scores \cite{gharibans2019spatial,agrusa2022robust}. In fields such as chemistry, uncovering interactions in chemical oscillators can reveal mechanisms underlying reaction-diffusion systems \cite{reaction_diffusion}.

The aforementioned tasks require tools capable of modeling the inherent complexity of interactions in multivariate phase relationships, including their spatial patterns and the dependencies between oscillators. To help move towards a solution to this problem, we state two desiderata when it comes to the design of a probability distribution:

\begin{enumerate}
    \item a statistical characterization that captures the spatial patterns of wave-like phenomena
    \item a description of the important interactions between oscillators.
\end{enumerate}

To address our first desideratum, we approach the problem by modeling wave-like phenomena through multivariate phase relationships. Since phase is a circular random variable defined on \([-\pi, \pi)\), this naturally leads to the consideration of phenomena on the \(n\)-torus. Statistical modeling of multivariate circular distributions on the \(n\)-torus has become increasingly important in fields such as bioinformatics and climatology, where it is used to analyze complex phenomena like protein structures and wind directions \cite{protein_pred,wind_direction}. 

In the univariate case, circular distributions, such as the von Mises distribution or wrapped analogs of distributions with support on the real line, such as Gaussians, are commonly used. Building on this foundation, recent advancements have introduced multivariate extensions, including the multivariate von Mises and multivariate wrapped distributions, which offer powerful tools for modeling interdependent circular data \cite{MvM,mGvM,wrapped_dist_review}.

Approaches to our second desideratum can be found in the larger field of graphical model inference. This is because  interactions between oscillators can be viewed as interactions between two circular random variables in a graph. This allows us to view the whole problem as a probabilistic graphical model. Approaches to learning graphical models can date back all the way to maximum likelihood estimation of multivariate distributions. 
One approach for instance that has been attractive within the neuroscience functional imaging community is estimating the inverse covariance matrix  of multivariate Gaussians.  In this context, the $(i,j)$ entry of the inverse covariance matrix being zero encodes conditional independence between variable $i$ and $j$ given all other nodes, thus eliminating indirect effects  \cite{huang2010learning,smith2011network}. The identification of a sparse set of interactions is built upon the more recent work involving $\ell_1$ regularized estimation. Most notably, the graphical lasso solves this problem in the cases of multivariate Gaussians \cite{friedman2008sparse}.
This work has been extended to Markov random fields and Bayesian networks more generally. For example, work by Vuffray et al, propose an interaction screening objective to estimate sparse Ising models from data in a sample efficient manner, and work by Shah et al, has also extended this estimator to bounded continuous real random variables whose pairwise interaction terms can be described by a Kronecker product \cite{vuffray_ising,Shah_Shah_Wornell_2021}. Within the specific context of multivariate phase interactions, Cadieu and Koepsell proposed a distribution and method to solve this problem using a score function methodology; however, they propose an $\ell_2$ regularized estimator and neither show theoretical guarantees on structure learning  nor sample complexity\cite{cadieu2010phase}.

In this work, we propose a probabilistic model of wave-like phenomena that allows for both hypothesis testing and structure learning. We use two methods of approximating/fitting the joint distribution on phase data: 1) a directed tree approximation and 2) an interaction screening approach first described  in the case of Ising models \cite{Chow_Liu_1968,vuffray_ising}. Our approach for analysis of the ``interaction screening objective'' (ISO) generalizes the Ising model, using circular symmetries and properties of minimal exponential families to guarantee theoretical guarantees  in terms of both structure learning and sample complexity. Our analysis extends the work by  also extend analysis by Shah et al. by considering edge potentials between node $i$ and node $j$ that are of the form $\cos(y_j-y_i)$, and thus cannot be written as a Kronecker product of terms only in $y_j$ and $y_i$ \cite{Shah_Shah_Wornell_2021}.  

In \cref{sec:statmodel}, we give preliminary definitions and develop the statistical model of multivariate phase relationships we use throughout the the manuscript.  In \cref{sec:treeApproximation}, we consider the optimal tree approximation and how it can be efficiently estimated with closed form expressions.  In \cref{sec:ISO}, we consider the more general problem of inferring the more general problem for an arbitrary graph, and discuss $\ell_1$ regularized convex optimization using ``interaction screeing'' approaches that can efficiently perform structure learning and have attractive sample complexity properties.     In \cref{sec:experiments}, we conduct controlled experiments where different spatiotemporal patterns are generated and we show the method's ability to identify the patterns and perform successful inference with log likelihood ratio tests for classification. In \cref{sec:discussion}, we conclude with a discussion and future steps.

\section{Statistical Modeling of Multivariate Phase Relationships}
\label{sec:statmodel}

\subsection{Preliminaries}
\newcommand{\vw}{\underline{w}}
\newcommand{\vv}{\underline{v}}
We begin by defining some basic quantities that we will use in this paper. We denote a length-$d$ vector as $\vv \in \reals^d$.
We denote the inner product between $\vw,\vv \in \reals^d$ as $\innerprod{\vw}{\vv} = \vw^T \vv$.  For $\vv \in \reals^d$, define  $\|\vv\|_\infty = \max_{i=1,\ldots, d} |\vv_i|$, \quad $\vv\|_2 = \sqrt{\sum_{i=1}^d v_i^2}$.

We will  use the following trigonometric identity throughout the manuscript:
\beqa
\cos(a-b)&=&\cos a \cos b + \sin a \sin b \label{eqn:trigIdentity:a}
\eeqa

Let $P$ and $Q$ be two distributions with support over the circle, $\cY=[-\pi,\pi)$, and have probability density functions $f(y)$ and $g(y)$ respectively.

\begin{definition}[Relative Entropy]
    The relative entropy from $P$ to $Q$ is given as
    \begin{align}
        D(P\|Q) &=  \E_P\brackets{{\log{\frac{f(Y)}{g(Y)}}}} \nonumber\\
        &= \int_{\cY} f(y){\log{\frac{f(y)}{g(y)}}} dy.
    \end{align}
\end{definition}
While the relative entropy is not a distance since it is not symmetric, it is often used a measure of how "far" two probability distributions are from each other. It can also be interpreted as the expected cost of miscoding a sequence of random variables using distribution $Q$, when the random variables are actually drawn according to distribution $P$ \cite{cover1999elements}.

Now consider a pair of random variables, $X$ and $Y$ that have joint density $P_{X,Y}$ and are marginally distributed according to $P_{X}$ and $P_{Y}$ respectively. 

\begin{definition}[Mutual Information]
    Consider random variables $X,Y$ with joint density $f_{X,Y}(x,y)$.  The mutual information between $X$ and $Y$ is defined as follows:
    \begin{subequations}
    \begin{align}
        I(X;Y) &= \int_x D(P_{Y|X=x} \| P_Y) f_X(x) dx  \label{eqn:defn:MutualInformation:a}\\
               &=\int_{x,y} f_{X,Y}(x,y) \log\frac{f_{X,Y}(x,y)}{f_X(x)f_Y(y)} dx dy. \label{eqn:defn:MutualInformation:b}
    \end{align} \label{eqn:defn:MutualInformation}
    \end{subequations}
\end{definition}
It is often interpreted as the amount of information gained about $Y$ by observing $X$ and vice versa \cite{cover1999elements}.

\begin{definition}[Restricted Strong Convexity]
    Consider a function $f(x)$ with minimizer $x^*$. Define 
  $\delta f(\Delta)$ as the error in the Taylor expansion at $x^*$:
    \begin{align}
    \delta f(\Delta) &\triangleq f(x^* + \Delta) - f(x^*) - \innerprod{\nabla f(x^*)}{\Delta}. \label{eqn:defn:Taylorerror}
    \end{align}    
    Then $f$ is restricted strongly convex with respect to a set $K$, if on a ball of radius $R$, for all $\Delta \in K$ such that $\norm{\Delta}_2 \leq R$, there exists some constant $\gamma$ such that
    \begin{eqnarray*}
        \delta f(\Delta) \geq \gamma \norm{\Delta}_2^2.
    \end{eqnarray*}
\end{definition}

\begin{definition}[Exponential Family]
    An exponential family is a parametric set of distributions such that for a parameter vector $\utheta = [\theta_1, \cdots, \theta_p]$, the distribution over the random vector $\underline{Y}$ can be written in the form
    \begin{align*}
        f_{\underline{Y}}(\underline{y};\utheta) = h\parenth{\underline{y}}\exp\parenth{
        \innerprod{\eta\parenth{\utheta}}{t\parenth{\underline{y}}}-A\parenth{\utheta)}},
    \end{align*}
    where $t\parenth{\underline{y}}$, $\eta\parenth{\utheta}$, $h(x)$, and $A\parenth{\utheta}$ are known functions and $h\parenth{\underline{y}}$ is nonnegative. The function $\eta\parenth{\utheta}$ is known as the natural parameters of the exponential family. \label{def:exp-fam}
\end{definition}

\begin{definition}[Minimal Exponential Family]
    An exponential family is minimal if there does not exist a vector $\underline{a}$ for which
    \begin{align*}
        \sum_{\alpha \in \mathcal{I}} a_{\alpha}\phi_{\alpha}(y) = c,
    \end{align*}
    where $\phi_{\alpha}(y)$ is the sufficient statistic, $\mathcal{I}$ is the index set that indexes the natural parameters of the exponential family model, and $c$ is a constant. \label{def:min-exp-fam}
\end{definition}

\begin{definition}[Bernstein's Inequality]
Let $X_1, \ldots, X_n$ be independent zero-mean bounded random variables with $\abs{X_i} \leq M$. Let $t$ be any positive number. Then the tail probabilities can be bounded as:
\begin{align*}
    \prob{\sum_{i=1}^n X_i \geq t} \leq \exp\parenth{-\frac{\frac{1}{2}t^2}{\sum_{i=1}^n \E\brackets{X_i^2} + \frac{1}{3}M t}}.
\end{align*}
\label{defn:bernstein}
\end{definition}

\subsection{Descriptions of Spatial Patterns of Waves}
In the study of waves, a natural starting point is the well known wave equation in physics given by
\begin{align}
    \frac{\partial^2 u}{\partial t^2} = \nabla^2 u, \label{eq:waveeq} 
\end{align}

where $u$ is a function of both space and time (e.g., a voltage waveform given by $u(x,y,t)$). This fundamental equation involves relating derivatives in space to derivatives in time, which implies spatiotemporal differences are very important to the characterization of a wave. Thus, for our probabilistic model it will be very important that we can characterize differences of this type. To help us understand this, we consider eigenfunction solutions to the wave equation. It is well known that eigenfunctions can take the form of complex exponentials, whose real and imaginary components are simple sinusouids. Consider the following solution to the wave equation
\begin{align}
    u(x,y,t) &= e^{j\phi(x,y,t)} \label{eq:sinewave}\\
    \phi(x,y,t) &= K_x x + K_y y - \omega t \label{eq:phi}
\end{align}
where $K_x$ and $K_y$ are wave numbers in the x and y directions, and $\omega$ is the angular frequency in time. While \eqref{eq:waveeq} operates on this whole solution, note that all of the description of the spatiotemporal patterns happens with the linearity of spatiotemporal phase relationships. Specifically, from \eqref{eq:phi}, we can immediately see that this linear relationship relationship is key to describing the spatiotemporal dynamics of plane waves. Moreover, given the gradient of the phase, 
\begin{align*}
    \nabla \phi(\cdot,\cdot,t) = \begin{bmatrix}
        K_x & K_y
    \end{bmatrix}^\top,
\end{align*}
we can understand what the relevant spatial patterns in the wave are at any given moment in time. 
Now suppose we are recording a wave at discretely sampled locations, $\{(x_i,y_i)\}_{i=1}^p$. Given a wave described by Eq. \eqref{eq:sinewave}, it now suffices to understand the discrete analog to the derivatives - the phase differences. Thus, we must have an explicit representation of the phase differences in our model given that our data is limited to a finite set of measurement locations.

\subsection{Interactions between Oscillators and Conditional Independence}
\label{sec:condindep}

Our second desiderata leaves us with the question: how can we describe which oscillators are interacting with which other oscillators? The more general form of this question is actually of great importance to those studying graphical models and is known to some as 'structure learning' \cite{wainwright2008graphical}. In this problem, we consider a probabilistic graphical model where we have some graph, $\mathcal{G} = (V,E)$ with distribution $P$, where $V$ is the vertex/node set containing our random variables, $E$ is the edge set describing interactions between nodes in $V$, and $P$ is a joint distribution over all nodes in $V$. Here, the edge set is defined by the conditional dependencies between random variables in the graph. That is, if node $i$ is conditionally independent of node $j$ given all other nodes, then the edge $(i,j) \notin E$ \cite{wainwright2008graphical}. This makes intuitive sense in the case of oscillators because if oscillator $i$ and $j$ share no more in common than the information from their other neighbors, then the two oscillators should not be interacting with each other.

It is important to note that the spatial phase relationship given by our first desiderata is not sufficient to capture the question of interactions. If a coupled phase phenomena were to follow the relationship implied in \eqref{eq:sinewave}, then it would be the case that every oscillator is equally as informative as any other oscillator. Thus, the question of interaction between oscillators would reveal no interesting relationship. However, in physical systems where it is known that there exist direct connections between oscillators, such as between different populations of neurons in the brain, but those connections are unknown, conditional independence can help provide a lens into that interaction.

\subsection{A Pairwise Exponential Family Distribution for Waves}

Given our two desiderata, we take additional inspiration from the graphical models community and seek a distribution that can encode both the phase offsets and conditional independencies in its parameters. We consider the following distribution:
\begin{eqnarray}
    f_{\vY}(\vy) &=& \frac{1}{Z} \exp \parenth{\sum_{(i,j) \in E} \kappa_{ij}\cos(y_j-y_i-\mu_{ij})}, \label{eq:jointdensity}
\end{eqnarray}
where $\kappa_{ij}$ is the coupling parameter between oscillators $i$ and $j$, $\mu_{ij}$ is the phase offset between oscillators $i$ and $j$, and $Z$ is the partition function that normalizes the distribution to integrate to one. 
We note that the distribution in \eqref{eq:jointdensity} can be interpreted as the equilibrium distribution of a Kuramoto system of inhomogeneous coupled oscillators where oscillator $i$ and oscillator $j$ have coupling $\kappa_{i,j}$ and inhomogeneous phase offsets given by $\mu_{i,j}$ with zero-mean stochastic oscillatory fluctuations \cite{cadieu2010phase}.  
The consideration of phase offsets appear to be an important feature of models, pertaining to axon conduction delays, in order to relate these coupled oscillations to traveling waves in the brain \cite{muller2021algebraic,davis2021spontaneous,budzinski2023analytical,gonzales2025touch}.

We now consider the conditional independence properties of this distribution as desired in \Cref{sec:condindep}. 

\begin{lemma}
    For any random vector distributed according to \eqref{eq:jointdensity}, $\kappa_{ij} = 0$ is equivalent to conditional independence between $Y_i$ and $Y_j$ given all nodes except for $i$ and $j$ \label{lemma:condindep}
\end{lemma}
\begin{proof}
    For any integer $u \in V$, define $\vy_{-u}$  to be the vector of observations on nodes excluding node $u$: $\vy_{-u} ={y_j: j \in V \setminus \braces{u}}$. By definition, if $Y_i$ and $Y_j$ are conditionally independent, the full joint density can be factored as
    \begin{eqnarray*}
        f_{\vY}(\vy) = h(\vy_{ -i})g(\vy_{-j}),
    \end{eqnarray*}
    where $h$ and $g$ are functions of all components besides $i$ and $j$ respectively. Letting $\kappa_{ij} = 0$, and noting that \eqref{eq:jointdensity} can be rewritten as
    \begin{eqnarray*}
        f_{\vY}(\vy) &=& \frac{1}{Z} \prod_{(i,j)\in E} \exp \parenth{\kappa_{ij}\cos(y_j-y_i-\mu_{ij})},
    \end{eqnarray*}
    we see that the only cross-term between $Y_i$ and $Y_j$ is eliminated and thus the joint distribution can be factored in a conditionally independent form. The other direction can be shown by contradiction.
\end{proof}
Thus, by \Cref{lemma:condindep}, estimation of the $\kappa_{ij}$ will allow us to infer the structure of the graphical model, while on the other hand, estimation of the $\mu_{ij}$ will give us information about the spatial pattern of phases.

This distribution actually generalizes the classical XY model in statistical physics, commonly used to model continuous magnetization on a lattice \cite{romano2002xy}. Specifically, by adding in the offset terms, $\mu_{ij}$, we allow for out-of-phase synchronization that is not in the classical XY model and can be applicable to a more general class of phase relationships \cite{cadieu2010phase}. We note that the model given by \eqref{eq:jointdensity} is also of the exponential family and is a product of pairwise von-Mises interaction terms, where the von-Mises density is given by
\begin{align}
    f_Y(y) = \frac{1}{2\pi I_0(\kappa)}\exp\parenth{\kappa\cos(y-\mu)} \label{eqn:defn:vonMises}
\end{align}
and $I_0(\cdot)$ is a modified Bessel function of order zero. We also observe that by \Cref{def:exp-fam}, the model specified by \eqref{eq:jointdensity} is of the exponential family, by using \eqref{eqn:trigIdentity:a}:
\begin{align}
    f_{\vY}(\vy; \theta) &= \frac{1}{Z(\utheta)} \exp \parenth{\sum_{(i,j) \in E} \thetaijc\cos(y_j-y_i)+\thetaijs\sin(y_j-y_i)}, \label{eq:naturaljointdensity}
\end{align}
where $\thetasijc = \kappa_{ij}\cos(\mu_{ij})$ and $\thetasijs = \kappa_{ij}\sin(\mu_{ij})$ are the model's natural parameters, $h\parenth{\underline{y}} = 1$, $A\parenth{\utheta} = \log\parenth{Z}$, and the sufficient statistic $T\parenth{\underline{y}}$ is the vector of $\sin$'s and $\cos$'s of pairwise differences.  This transformation will make our inference problem easier to solve.  
While this distribution satisfies our two desiderata, solving for the maximum likelihood parameters of this distribution is well known to be a difficult problem. In particular, with samples $\vy^{(1)}, \ldots, \vy^{(n)}$, we would have to solve the following maximum likelihood estimation problem:
\begin{eqnarray*}
\hat{\utheta} &=& \arg\min_{\utheta} \;\;\frac{1}{n} \sum_{k=1}^n -\log f_{\vY}(\vy^{(k)}; \theta) \nonumber \\
 &=& \arg\min_{\utheta} \;\;\log Z(\utheta) + \frac{1}{n}\sum_{k=1}^n \sum_{(i,j)\in E} \thetaijc \cos(\yk_j-\yk_i) +\thetaijs \sin(\yk_j-\yk_i), \label{eq:MLE}
\end{eqnarray*}

where the partition function $Z(\utheta)$ is given by
\begin{align*}
    Z(\utheta) = \int_{\vy \in \cY^p } \exp\parenth{\sum_{(i,j)\in E} \thetaijc \cos(y_j-y_i) +\thetaijs \sin(y_j-y_i)} d\vy
\end{align*}

We see here that the partition function is not only a high-dimensional integral, but also a function of the parameters, which makes this optimization problem very computationally expensive. In the following section, we will introduce both the Chow-Liu  and interaction screening methods to either approximate the distribution or efficiently estimate its parameters. 

\subsection{Useful Properties of the Multivariate Phase Model}
\subsubsection{Minimal Exponential Family Property}
\label{sec:minimalexpfam}

In this section, we show the following lemma:

\begin{lemma}
    The model given in \eqref{eq:naturaljointdensity} is a minimal exponential family. \label{lemma:min-exp-fam}
\end{lemma}

\begin{proof}

First note that the full vector sufficient statistic for our model is
\begin{align}
    \underline{\phi}(y) = \begin{bmatrix}
        \underline{\phi_c}(y)\\
        \underline{\phi_s}(y)
    \end{bmatrix},
\end{align}
where for $(i,j) \in E$
\begin{align}
    \phi_{ij,c}(y) &= \cos(y_j-y_i)\\
    \phi_{ij,s}(y) &= \sin(y_j-y_i).
\end{align}
Thus, in accordance with \Cref{def:min-exp-fam}, we desire to show that the following quantity
\begin{align}
    \sum_{(i,j)\in E} a_{ij,c}\cos(y_j-y_i) + a_{ij,s}\sin(y_j-y_i), \label{eq:minimalip}
\end{align}
cannot equal a constant for all $\vy \in \mathcal{Y}^p$ for any possible $\underline{a}$ vector of coefficients for this model to be of a minimal exponential family.

Now suppose that there did exist an $\underline{a}$ for which \eqref{eq:minimalip} is constant. This means that a perturbation in any component of $\vy$ should not change the result of the inner product. Without loss of generality, let us only perturb $y_1$. We then focus on the portion of the sum that involves the edges involving node $1$. That is,
\begin{align*}
    \sum_{(i,j)\in E_1} a_{ij,c}\cos(y_j-y_i) + a_{ij,s}\sin(y_j-y_i) = \sum_{j = 2}^p a_{1j,c}\cos(y_j-y_1) + a_{1j,s}\sin(y_j-y_1)
\end{align*}
Since all of these terms are sinusoids that are a function of $y_1$ alone (since we hold all other components constant), then we can rewrite this sum as a single resultant sinusoid with some magnitude and phase:
\begin{align*}
    \sum_{j = 2}^p a_{1j,c}\cos(y_j-y_1) + a_{1j,s}\sin(y_j-y_1) = A(y_{-1})\cos(y_1 - \psi(y_{-1})),
\end{align*}
where the magnitude and phase shift depend on all components except $y_1$. Since this term is sinusoidal, unless the magnitude is always zero, it cannot be constant. Thus, we just need to show that $A(y_{-1})$ is not always zero. By basic trigonometry, we can show that
\begin{align*}
    \sum_{j = 2}^p a_{1j,c}\cos(y_j-y_1) + a_{1j,s}\sin(y_j-y_1) &= \sum_{j = 2}^p a_{1j,c}\cos(y_1-y_j) + a_{1j,s}\sin(y_1-y_j + \pi)\\
    &= \sum_{j = 2}^p A_j(y_j)\cos\parenth{y_1 + \psi_j(y_j)},
\end{align*}
where
\begin{align*}
    A_j(y_j) &=\sqrt{\parenth{a_{1j,c}\cos(y_j) + a_{1j,s}\sin(\pi - y_j)}^2 + \parenth{a_{1j,c}\sin(y_j)  -a_{1j,s}\cos(\pi - y_j)}^2}\\
    \psi_j(y_j) &= \arctan{\parenth{\frac{a_{1j,c}\sin(y_j)  -a_{1j,s}\cos(\pi - y_j)}{a_{1j,c}\cos(y_j) + a_{1j,s}\sin(\pi - y_j)}}}.
\end{align*}
Since each of the $A_j(y_j)$ depend on $y_j$ independently and cannot be zero for all $y_j$, as long as $\underline{a} \neq 0$, we conclude that $A(y_{-1})$ cannot be always zero if $\underline{a} \neq 0$. Thus \eqref{eq:jointdensity} is a minimal exponential family.
\end{proof}

\subsubsection{Univariate Conditional Distributions of the Model are von-Mises} \label{subsec:univariate_conditional}

Univariate conditionals of joint distributions are often very useful - aiding in the understanding the statistics of a single random variable and even in sampling from difficult distributions through Markov chain Monte Carlo methods. Here we show that the univariate conditionals of the model described in \eqref{eq:jointdensity} are simply von-Mises distributions. 

To understand the univariate conditionals, we first define the edge set of all edges associated with node $u$ to be 
$$E_u \triangleq \{(i,j) \in E \;\big|\; i = u \text{ or } j = u\}$$.  

With that, we have the following lemma:

\begin{lemma}\label{lemma:univariate:Conditional}
For any $u \in E$, the conditional distribution of $Y_u$ given all other nodes is a von Mises distribution.
\end{lemma}

\begin{proof}
By Baye's Rule, note that
\begin{align*}
    f_{Y_u \mid Y_{-u}}(y_u \given{\big} y_{-u}) \propto f_{\vY}(\vy).
\end{align*}
From this we can deduce the following,
\begin{align*}
    f_{Y_u \mid Y_{-u}}(y_u \given{\big} y_{-u}) &\propto \exp \parenth{\sum_{(i,j) \in E} \thetasijc \cos(y_j-y_i) +  \thetasijs \sin(y_j-y_i)}\\
    &=\exp \parenth{\sum_{(i,j) \in E_u} \thetasijc \cos(y_j-y_i) +  \thetasijs \sin(y_j-y_i)}\\
    &\;\;\;\times \exp \parenth{\sum_{(i,j) \in E_u^c} \thetasijc \cos(y_j-y_i) +  \thetasijs \sin(y_j-y_i)}\\
    &=C_{-u} \exp \parenth{\sum_{(i,j) \in E_u} \thetasijc \cos(y_j-y_i) +  \thetasijs \sin(y_j-y_i)}
\end{align*}
Letting $\kappa_{ij} = \sqrt{{\thetasijc}^2 + {\thetasijs}^2}$ and $\mu_{ij} = \arctan\left(\frac{\thetasijs}{\thetasijc}\right)$,
\begin{align}
    f_{Y_u \mid Y_{-u}}(y_u \given{\big} Y_{-u} = y_{-u}) &\propto C_{-u} \exp \parenth{\sum_{(l,k) \in E_u} \kappa_{lk} \cos(y_k-y_l-\mu_{lk})}\nonumber\\
    &=  C_{-u} \exp \parenth{\sum_{k\in V\backslash u} \kappa_{uk} \cos(y_k-y_u-\mu_{uk})} \label{eq:cos-sym}\\
    &= C_{-u} \exp \parenth{ \Re\left( 
     \sum_{k\in V\backslash u} \kappa_{uk} \exp\left(i(y_k-y_u-\mu_{uk})\right) \right)}\nonumber\\
    &= C_{-u} \exp \parenth{ \Re\left( \exp\left(-iy_u \right)
     \sum_{k \in V\backslash u} \kappa_{uk} \exp(i(y_k-\mu_{uk}))\right)} \nonumber\\
     &= C_{-u} \exp \parenth{ \Re\left( \exp\left(-iy_u \right)
     A\exp(i\theta_u)\right)} \nonumber\\
     &= C_{-u} \nonumber\exp\parenth{A\cos\parenth{\theta_u-y_u}}\\
     &=  C_{-u} \exp\parenth{A\cos\parenth{y_u-\theta_u}}, \nonumber
\end{align}
where \eqref{eq:cos-sym} follows from the symmetry of cosines, $A = \abs{\sum_{k \in V\backslash u} \kappa_{uk} \exp(i(y_k-\mu_{uk}))}$, and $\theta_u = \arctan\left(\sum_{k \in V\backslash u} \kappa_{uk} \exp(i(y_k-\mu_{uk}))\right)$. This is  the form of a von Mises density by virtue of \eqref{eqn:defn:vonMises}.
\end{proof}

\subsubsection{Relative Entropy and Mutual Information of a Pair of Oscillators}

Consider the wave model in the case of two oscillators, i.e., $p=2$. We compute the pairwise mutual information under this model and show that under this formulation, we have a closed-form solution. We first show the following results on the two-dimensional model and then proceed to compute the information measures.
\begin{lemma}
    The marginal distribution on the two-dimensional wave model, of the form in \Cref{eq:jointdensity}, is uniform. \label{lemma:marginal-2d}
\end{lemma}
\begin{proof}
    
Suppose we have a model of the form in \eqref{eq:jointdensity}, with $p=2$. To find the marginal, on say $Y_1$, we simply integrate out $Y_2$, where for any $y_1 \in [-\pi,\pi)$:
\begin{align}
    f_{Y_1}(y_1) &= \int_{y_2 = -\pi}^\pi f_{Y_1,Y_2}(y_1,y_2) dy_2\nonumber\\
    &=  \int_{y_2 = -\pi}^\pi \frac{1}{
    Z(\kappa_{12},\mu_{12})} \exp\parenth{\kappa_{12}\cos\parenth{y_2-y_1-\mu_{12}} }dy_2\nonumber\\
    &= \frac{2\pi I_0(\kappa_{12})}{Z(\kappa_{12},\mu_{12})} \label{eq:p2_marginal} \\
    &=\frac{1}{2\pi}, \label{eq:p2_marginal:a} 
\end{align}
where \eqref{eq:p2_marginal} results from the integral of the density of a univariate von-Mises (since $y_1$ is held constant); and
\eqref{eq:p2_marginal:a} follows since the density in \eqref{eq:p2_marginal}is a constant over $[-\pi,\pi)$. 
\end{proof}

\begin{corollary}
    The univariate conditional of the wave model in two-dimensions is von-Mises.
\end{corollary}
\begin{proof}
    By \Cref{lemma:marginal-2d}, we know that the marginal of such a two-dimensional distribution is uniform. Thus we have
    \begin{align}
        f_{Y_2\mid Y_1}(y_2\mid y_1) &= \frac{f_{Y_2, Y_1}(y_2, y_1) }{f_{Y_1}( y_1) } \nonumber\\
        &= \frac{1}{
    Z(\kappa_{12},\mu_{12})} \exp\parenth{\kappa_{12}\cos\parenth{y_2-y_1-\mu_{12}} } \times \parenth{\frac{2\pi I_0(\kappa_{12})}{Z(\kappa_{12},\mu_{12})} }^{-1} \label{line:2d-product} \\
    &= \frac{1}{2\pi I_0(\kappa_{12})}\exp\parenth{\kappa_{12}\cos\parenth{y_2-y_1-\mu_{12}} }, \label{eq:conditional-2d}
    \end{align}
    where \eqref{line:2d-product} follows from \eqref{eq:p2_marginal}.    
\end{proof}

Given the above construction, we consider the description of $I(X;Y)$ given by \eqref{eqn:defn:MutualInformation:a} where $P_Y$ will be the uniform distribution on $[-\pi,\pi)$.
Thus we have that for any $x$, the relative entropy is given by
\begin{eqnarray}
D(P_{Y|X=x}\|P_Y) &=& \int_{y=-\pi}^\pi 
f_{Y|X}(y|x) \log \frac{f_{Y|X}(y|x)}{f_Y(y)}dy\\
                  &=& \log(2 \pi) - h(P_{Y|X=x}) \\
                  &=& \log(2 \pi) + \kappa \frac{I_1(\kappa)}{I_0(\kappa)}
                  - \log\left(2 \pi I_0(\kappa)\right)  \nonumber\\
                  &=& \kappa \frac{I_1(\kappa)}{I_0(\kappa)} -\log I_0(\kappa) \label{eq:pairKL}.
\end{eqnarray}
where $h(\cdot)$ is the Shannon entropy, which for a von Mises distribution with parameter $\kappa$ has a closed form expression as $\log\left(2 \pi I_0(\kappa)\right) -\kappa \frac{I_1(\kappa)}{I_0(\kappa)}
                  $.  As such, the mutual information is given by
\begin{align}
I(X;Y) =\kappa \frac{I_1(\kappa)}{I_0(\kappa)} -\log I_0(\kappa). \label{eq:pairMI}
\end{align}

\section{Methods for Estimating Optimal Tree Approximations} \label{sec:treeApproximation}

\newcommand{\tree}{\mathcal{T}}
\newcommand{\cTrees}{\mathsf{T}}

\subsection{The Chow-Liu Method for Approximating the Joint Distribution}
The main goal of this work is to infer from data and extract the parameters of our random field model, as proposed in \eqref{eq:jointdensity}. However, learning a fully flexible joint distribution, is, in general, a hard problem. Therefore, we first introduce an information-theoretic approach to approximating the joint distribution with a tree structure, by using pairwise statistics. This approach was first introduced by its namesake, Chow and Liu, in 1968 for discrete alphabets and has been widely influential in applied probability and information theory \cite{Chow_Liu_1968}.

We begin by describing the dependence tree. Suppose we have a graph $G = (V,E)$ with four vertices, $V = \{1,2,3,4\}$ and some edge set, $E$. We say that the graph $G$ is a tree if $G$ is: 1) acyclic, meaning there are no cycles in the graph, and 2) connected, meaning that there are no disjoint subgraphs. Therefore, one valid edge set that would make this graph a tree would be $E = \{(1,2), (1,3), (3,4)\}$. This is demonstrated below in \Cref{fig:dependenceTree}.
\begin{figure}[htp]
    \centering    \includegraphics[width=\linewidth]{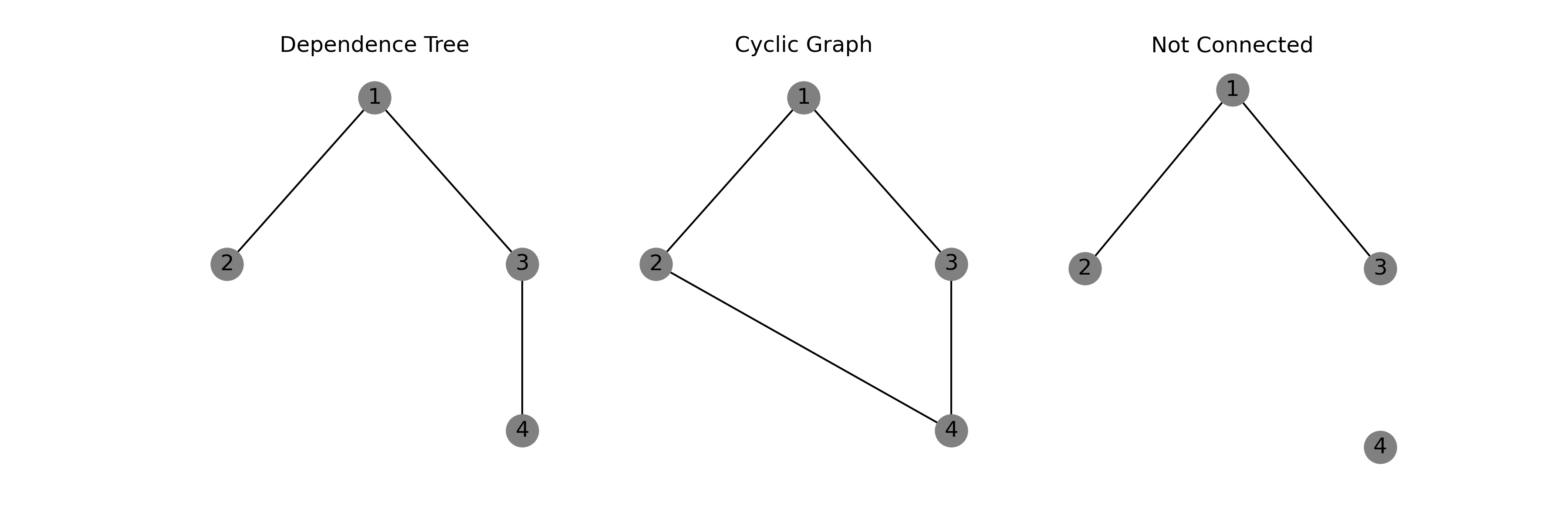}
    \caption{Example graphs depicting a dependence tree (left) and two graphs that are not dependence trees (middle, right), each violating a difference property. The middle graph violates the acyclic property of trees and the right graph violates the connected property of trees.}
    \label{fig:dependenceTree}
\end{figure}
To be a dependence tree, however, we must also encode information about the probability distribution. We do this using the following factorization of the joint distribution
\begin{align}
    f^\tree_{\vY}(\vy) = \prod_{i = 1}^d f_{Y_i|Y_{j(i)}}(y_i|y_{j(i)}), \label{eqn:defn:treedistribution}
\end{align}
where for tree $\tree$, $j(i)$ is the parent of node $i$ and if node $i$ has no parent (i.e., the root node) we use its marginal distribution. The tree structure and "parent-child" relationship of nodes falls out once one picks some node to be the root node (any node will work). Using the leftmost graph in \Cref{fig:dependenceTree}, if we choose the root to be node 1, then the distribution would factor as
\begin{align*}
    f^\tree_{\vY}(\vy) = f_{Y_1}(y_1)f_{Y_2|Y_1}(y_2|y_1)f_{Y_3|Y_1}(y_3|y_1)f_{Y_4|Y_3}(y_4|y_3).
\end{align*}

\subsubsection{The Chow-Liu Dependence Tree}

Let $\vY = [Y_1, Y_2, \cdots, Y_p]$ be a random vector in $[-\pi,\pi)^d$. That is, each random variable $Y_i$ has circular support. In general, the joint distribution can have any form, but here we limit ourselves to approximations of the maximum likelihood estimate from the class of tree distributions. We will state the main result of the Chow-Liu algorithm here and then derive necessary component parts specific to our method later. 

\begin{definition}
Let $\cTrees$ be the set of all trees with $p$ nodes.  For any $\tree \in \cTrees$, let $P^\tree$ be the tree distribution with density given by
\eqref{eqn:defn:treedistribution} and $P$ be the true generative distribution of the data, with density \eqref{eq:jointdensity} for some set of parameters $\underline{\theta}^*$. Let $I(Y_i;Y_j)$ be the pairwise mutual information between any two of these random variables under $P$. Then the solution to the minimization problem
\begin{equation}
    \hat{T} = \arg\min_{\tree \in \mathcal{T}} D(P\| P^\tree)
\end{equation}
is the distribution given by the maximum spanning tree, where edge values between $Y_i$ and $Y_j$ are given by the pairwise mutual information \cite{Chow_Liu_1968}. That is, 
\begin{equation}
    \hat{T} = \arg\max_{\tree \in \mathcal{T}} \sum_{i = 1}^n I(Y_i;Y_{j(i)}),
\end{equation}
where $I(Y_i;Y_{j(i)})$ represents the edge weight between $Y_i$ and some $Y_j$ connected to $Y_i$ with respect to tree $\tree$. This is known as the Chow-Liu dependence tree.
\end{definition}

Constructing the Chow-Liu dependence tree, thus, reduces the problem of finding a maximum likelihood estimate in the set of dependence trees to a maximization over pairwise mutual informations. 

\subsubsection{Chow-Liu Approximation of the Graphical Model}

In this section, we describe a natural way to parameterize the Chow-Liu dependence tree for the graphical model in \eqref{eq:jointdensity}. Given that the Chow-Liu procedure needs pairwise mutual informations, we need some way to compute bivariate distributions on pairs of random variables. Additionally, we would like for this formulation to be faithful to our proposed original distribution. To do this we describe the bivariate distribution using the law of conditional probability, 
\begin{align*}
    f_{Y_1,Y_2}(y_1,y_2) = f_{Y_1}(y_1)f_{Y_2|Y_1}(y_2|y_1),
\end{align*}
and give models for the marginal and conditional distributions shown. We assume that the marginal statistics of an oscillation are given by a uniform distribution and the conditionals are given as follows:

\begin{equation}
    f_{Y_i|Y_j}(y_i|y_j) = \frac{\exp{(\tilde{\kappa}_{ij}\cos(y_i-y_j-\tilde{\mu}_{ij})})}{2\pi I_0(\tilde{\kappa}_{ij})}, \label{eq:chowliucond}
\end{equation}

where, $\tilde{\kappa}_{ij}$ and $\tilde{\mu}_{ij}$ are the pairwise concentration and mean parameters, not to be confused with $\kappa_{ij} $ and $\mu_{ij}$ from the graphical model in \eqref{eq:jointdensity}. Given this formulation, we can construct the dependence tree by fitting univariate von Mises distributions on all $y_j-y_i$ to get all $\tilde{\kappa}_{ij}$, compute all $I(Y_i;Y_j)$ using the $\tilde{\kappa}_{ij}$'s and \eqref{eq:pairMI}, and then use a maximum spanning tree algorithm to obtain the dependence tree. The approximated joint density can then be computed as
\begin{align}
    \hat{f}_{\vY}(\vy) = f^{\hat{\tree}}_{\vY}(\vy) \prod_{(i,j)\in \hat{T} }^n f_{Y_{j}|Y_i}(y_{j}|y_i). \label{eq:chowliujoint}
\end{align}

\section{Interaction Screening for Inference of the Full Graphical Model} \label{sec:ISO}

In the previous section, we described a general procedure to efficiently infer a tree distribution on the data. However, it may not be the case that the data were generated from the class of tree distributions. Thus, we consider the more general problem of inference over all possible graphs. In particular, we are interested in two problems: 1) sparse structure recovery and 2) sample-efficient parameter estimation on the inferred structure of the graph proposed in \eqref{eq:jointdensity}.  Previous efforts proposed the same graphical model we propose in \eqref{eq:jointdensity}, but the proposed estimation procedure using a score function methodology with $\ell_2$ regularization  was neither assessed from a structure learning nor sample complexity perspective \cite{cadieu2010phase}.

Recent work has shown that parameters for the Ising model, which models pairwise interactions between Rademacher random variables, can be solved for efficiently by convex estimation using the interaction screening objective (ISO) with guarantees on sample complexity \cite{vuffray_ising}. They later extended this work to include a general framework for discrete random variables. Work by Shah et al. has shown information theoretically good bounds on a similar problem for continuous random variables of the exponential family, however, the pairwise interaction terms are restricted to those that can be represented as a Kronecker product of functions on each individual term \cite{Shah_Shah_Wornell_2021}. Here we extend the interaction screening framework used originally for Ising models to continuous circular random variables. We show for the specific case of von-Mises random fields, whose pairwise potentials are not of the form in Shah et al, that we can solve the $\ell_1$ regularized structure learning problem efficiently \cite{Shah_Shah_Wornell_2021}.

We define the ISO for each node $u \in \{1,\cdots,p\}$ with $n$ i.i.d. samples as follows:

\beqa
\ISO &=& \frac{1}{n} \sum_{k=1}^n \exp \parenth{-\sum_{(i,j) \in E_u} \thetaijc \cos(\yk_j-\yk_i) +\thetaijs \sin(\yk_j-\yk_i)}, \label{eq:ISO}
\eeqa
where $E_u$ is the edge set constrained to only edges involving node $u$ as defined above. 
Compared to the maximum likelihood estimation (MLE) problem in \eqref{eq:MLE}
which contains the partition function, but solves for all parameters simultaneously, the ISO solves for parameters associated only with node $u$, but avoids having to deal with the computational complexity of the partition function. In particular, we can now solve the following problem:
\begin{align}
    \hat{\utheta}_u = \arg\min_{\uthetau} \;\; \ISO
    \label{eqn:ISO-problem}
\end{align}
for each node. For the structure learning problem, we can solve the $\ell_1$ regularized version of the problem as follows:
\newcommand{\ukappa}{\underline{\kappa}}
\begin{align}
    \hat{\utheta}_u[\lambda] = \arg\min_{\uthetau} \;\; \ISO+ \lambda \norm{\ukappa_u}_1, \label{eq:sparse-iso-problem}
\end{align}
where $\kappa_u$ is the vector of coupling parameters involving node $u$ before reparametrization: for any edge $(i,j) \in E_u$, $\kappa_{i,j} = \sqrt{\utheta_{i,j,c}^2+\utheta_{i,j,s}^2}$. We penalize the $\ukappa$'s rather than the $\theta$'s because in the structure learning problem we aim to look for sparsity in the parameters that correspond to conditional independence. By Lemma \eqref{lemma:condindep}, sparsity on the $\ukappa$'s achieves the desired behavior.

\subsection{Conditions for Controlling the Error in Parameter Estimates}

We follow the proof technique of the original interaction screening approach \cite{vuffray_ising}, which leverages general conditions for convergence of high dimensional M-estimators \cite{Negahban_Ravikumar_Wainwright_Yu_2012}. We first state the two conditions  and then the proposed theorem that quantitatively shows the utility of this method.

\begin{condition}
    Let $\lambda$ be the penalty parameter for an $\ell_1$ regularized M-estimator described by \eqref{eq:sparse-iso-problem}. $\lambda$ is greater than twice the magnitude of the largest partial derivative of the ISO:
    \begin{eqnarray*}
        \lambda \geq 2\norm{\nabla S_n}_\infty.
    \end{eqnarray*} \label{cond:1}
\end{condition}

As noted in Vuffray et al, this condition allows for control of the estimation error between the regularized solution and the true parameters to lie in a set
\begin{eqnarray}
K := \{\Delta \in \reals^{p-1} \mid \norm{\Delta}_1 \leq 4\sqrt{d} \norm{\Delta}_2 \} \label{eqn:defn:setK}
\end{eqnarray} 
where $\Delta$ is a deviation from the minimizer of the ISO and $d$ is the maximum number of non-zero entries in the true $\utheta^*$ \cite{vuffray_ising}.

\begin{condition}
    The ISO objective in
    \eqref{eqn:ISO-problem} is restricted strongly convex with respect to $K$ as defined in \eqref{eqn:defn:setK} on a ball of radius $R$. \label{cond:2}
\end{condition}

This condition allows for control of the minimizer of the ISO objective to lie close to the true parameter vector with at most $d$ non-zero entries.

\begin{proposition}
    If \Cref{cond:1} and \Cref{cond:2} are satisfied for some $\ell_1$ regularized $M$-estimator, and we have that $R > 3 \sqrt{d}\frac{\lambda}{\gamma}$, with $d$ being the maximum degree of a node, then the norm error in that regularized estimator is controlled as
    \begin{eqnarray}
        \norm{\hat{\utheta} - \utheta^*}_2 \leq 3 \sqrt{d}\frac{\lambda}{\gamma},
    \end{eqnarray}
    where $\utheta^*$ is the true parameter vector for the model from which the data are generated and $\hat{\utheta}$ is the $\ell_1$ regularized M-estimator in \eqref{eq:sparse-iso-problem}.
\end{proposition}
We also let
\begin{align}
    \min_{(i,j)\in E} \kappa_{ij} = \alpha,\quad
    \max_{(i,j)\in E} \kappa_{ij} = \beta. \label{eqn:defn:alpha:beta}
\end{align}
This will become useful in doing the analysis for these conditions on our model and our theorems. We continue with the rest of this section by following an analogous proof strategy to  the Ising model ISO approach \cite{vuffray_ising}. We introduce generalizations to the circular domain to bound moments of the gradient of the ISO and take advantage of the properties of minimal exponential families to show strong convexity. 

\subsection{Controlling the Gradient of the ISO}
To acheive condition 1, we need to control the gradient of \eqref{eq:ISO}.

For the sake of analysis, let
\begin{align}
    X_{ij,c}(\uthetau) &= -\cos(y_j-y_i) \exp \parenth{-\sum_{(l,m) \in E_u} \theta_{lm,c} \cos(y_m-y_l) +\theta_{lm,s} \sin(y_m-y_l)} \label{eqn:defn:Xijc} \\
    X_{ij,s}(\uthetau) &= -\sin(y_j-y_i) \exp \parenth{-\sum_{(l,m) \in E_u} \theta_{lm,c} \cos(y_m-y_l) +\theta_{lm,s} \sin(y_m-y_l)} \label{eqn:defn:Xijs}.
\end{align}
Thus the gradient, for any $(i,j) \in E_u$, is given by as
\begin{align}
\begin{matrix}
    \frac{\partial}{\partial \thetaijc }\ISO = \frac{1}{n}\sum_{k=1}^n \Xkijc(\utheta_u), &
    \frac{\partial}{\partial \thetaijs }\ISO = \frac{1}{n}\sum_{k=1}^n \Xkijs(\utheta_u)
\end{matrix} \label{eq:iso_gradient}
\end{align}
for the sine and cosine component respectively. We next aim to use a concentration inequality to bound 
$\norm{\nabla S_n(\uthetau^*)}_\infty$,
as described in \eqref{eq:iso_gradient}, to bound its behavior with high probability. To do this, we will use Bernstein's inequality, given in \Cref{defn:bernstein}, on each component (showing only the cosine component as example) as follows:
\begin{align*}
    \mathbb{P}\parenth{\abs{\frac{\partial}{\partial \theta_{ij,c}} S_n(\uthetau^*)} \geq t} \leq 2\exp\parenth{-\frac{\frac{1}{2}t^2n^2}{\sum_{k=1}^n \E\brackets{\parenth{\Xkijc}^2} + \frac{1}{3}M tn }}.\\
\end{align*}

We will next use the union bound to control $\norm{\nabla S_n(\uthetau^*)}_\infty$, by first characterizing the following equalities and inequalities for the $\sin$ and $\cos$ components:
\begin{enumerate}
    \item $\E\brackets{X_{ij,c}(\utheta_u^*)} = \E\brackets{X_{ij,s}(\utheta_u^*)} = 0$
    \item $\E\brackets{X_{ij,c}(\utheta_u^*)^2}+ \E\brackets{X_{ij,s}(\utheta_u^*)^2} = 1$
    \item 
    $\max\braces{\abs{X_{ij,c}(\utheta_u^*)},\abs{X_{ij,s}(\utheta_u^*)}}\leq M$ for quantity $M>0$ that we determine below
\end{enumerate}

Each of the above is shown in three lemmas below. First we provide some notation that will aid in making the proofs succinct.
\newcommand{\thetalmstar}{{\underline{\theta}_{l,m}^*}}
\newcommand{\tlmstar}{\underline{t}_{l,m}}

\newcommand{\sumlminSetinnerprod}[1]
{\sum_{(l,m) \in #1} \innerprod{\thetalmstar}{\tlmstar} }

\begin{definition}
    For any $(l,m) \in E$, define 
    \beqa
    \thetalmstar &\triangleq& \brackets{\theta^*_{lm,c},\theta^*_{lm,s}}^T \label{eqn:defn:thetalmstar}\\
    \tlmstar &\triangleq& \brackets{\cos(y_m-y_l),\sin(y_m-y_l)}^T \label{eqn:defn:tlmstar}
    \eeqa
    and note as a consequence that
    \[ \innerprod{\thetalmstar}{\tlmstar} = \theta^*_{lm,c} \cos(y_m-y_l) + \theta^*_{lm,s} \sin(y_m-y_l)\]
Moreover, from \eqref{eqn:defn:Xijc} and \eqref{eqn:defn:Xijs} it follows that
\begin{subequations}
\beqa
 X_{ij,c}(\uthetau^*) &=&-\cos(y_j-y_i)
 \exp \parenth{-
 \sumlminSetinnerprod{E_u}} \\
 X_{ij,s}(\uthetau^*) &=& -\sin(y_j-y_i)
 \exp  \parenth{-
 \sumlminSetinnerprod{E_u}} 
\eeqa \label{eqn:defn:Xij:innerprod}
\end{subequations}
and from \eqref{eq:naturaljointdensity} it follows that
\beqa
f_{\vY}(\vy; \utheta^*) &=& 
\frac{1}{Z(\utheta^*)} \exp   \parenth{
 \sumlminSetinnerprod{E}}
\label{eqn:defn:jointdensity:innerprod}
\eeqa    
\end{definition}

\begin{lemma}
$\E\brackets{\Xijc(\uthetau^*)} = \E\brackets{\Xijs(\uthetau^*)} =  0$ \label{lemma:firstmoment}
\end{lemma}

\begin{proof}
Since $(i,j) \in E_u$, assume without loss of generality that $j=u$. 
\begin{align}
\E\brackets{\Xijc(\utheta_u^*)} &= - \int_{\vy \in \cY^p} f_{\vY}(\vy) \cos(y_j-y_i) \exp \parenth{-\sumlminSetinnerprod{E_u}} d\vy\nonumber \\
&= - \int_{\vy \in \cY^p} f_{\vY}(\vy) \cos(y_j)\cos(y_i) \exp \parenth{-\sumlminSetinnerprod{E_u}} d\vy \nonumber \\
&\hspace{.35cm}- \int_{\vy \in \cY^p} f_{\vY}(\vy) \sin(y_j)\sin(y_i) \exp \parenth{-\sumlminSetinnerprod{E_u}} d\vy \label{eqn:proof:lemma:firstmoment:a}\\
&= - \frac{1}{Z} \int_{\vy \in \cY^p}  
\cos(y_u)\cos(y_i) \exp \parenth{\sumlminSetinnerprod{E}} \exp \parenth{-\sumlminSetinnerprod{E_u}} d\vy \nonumber \\
&\hspace{.35cm}
- \frac{1}{Z} \int_{\vy \in \cY^p}  
\sin(y_u)\sin(y_i) \exp \parenth{\sumlminSetinnerprod{E}} \exp \parenth{-\sumlminSetinnerprod{E_u}} 
d\vy\label{eqn:proof:lemma:firstmoment:b} \\
&= - \frac{1}{Z} \int_{\vy \in \cY^p}  \cos(y_u)\cos(y_i) \exp \parenth{\sumlminSetinnerprod{E \setminus E_u}} d\vy\nonumber \\
&\hspace{.35cm}- \frac{1}{Z} \int_{\vy \in \cY^p}  \sin(y_u)\sin(y_i) \exp \parenth{\sumlminSetinnerprod{E \setminus E_u}} d\vy \nonumber\\
&= - \frac{1}{Z} \left(\int_{-\pi}^\pi \cos(y_u) dy_u \right) \int_{\vy_{-u} \in \cY^{p-1}}  \cos(y_i) \exp \parenth{\sumlminSetinnerprod{E \setminus E_u}} d\vy_{-u} \nonumber\\
&\hspace{.35cm}- \frac{1}{Z} \left(\int_{-\pi}^\pi \sin(y_u) dy_u \right) \int_{\vy_{-u} \in \cY^{p-1}}  \sin(y_i) \exp \parenth{\sumlminSetinnerprod{E \setminus E_u}} d\vy_{-u} \label{eqn:proof:lemma:firstmoment:c}\\
&=0 \label{eqn:proof:lemma:firstmoment:d}.
\end{align}
where \eqref{eqn:proof:lemma:firstmoment:a} follows from \eqref{eqn:trigIdentity:a};  
\eqref{eqn:proof:lemma:firstmoment:b} follows from \eqref{eqn:defn:jointdensity:innerprod};
and the separation of the integral with $y_u$ in \eqref{eqn:proof:lemma:firstmoment:c} follows because the sum of edges in $E \setminus E_u$  clearly will not involve any edges containing node $u$.

$\E\brackets{\Xijs(\utheta_u)} = 0$ by a symmetric argument to above.
\end{proof} 

Note, here the similarity with the analogous proof for the Ising model
\cite[Lemma 1]{vuffray_ising}. In that context, they take advantage of the group structure by summing over the group $\{-1,1\}$, to arrive at an expected value of zero. Here, we perform a similar operation where we exploit how the integral of a sinusoidal function over a period in\eqref{eqn:proof:lemma:firstmoment:d} is zero.

\begin{lemma}
$\E\brackets{\Xijc(\uthetau^*)^2 + \Xijs(\uthetau^*)^2}=1$. \label{lemma:secondmoment}
\end{lemma}

\newcommand{\thetasabs}{}
\newcommand{\thetasabc}{}

\begin{proof}
\begin{align}
& \E\brackets{\Xijc(\utheta_u^*)^2 + \Xijs(\utheta_u^*)^2}  \nonumber \\
&= \E\brackets{ \parenth{\cos^2(Y_j-Y_i) +  \sin^2(Y_j-Y_i)}\exp \parenth{-2\sumlminSetinnerprod{E_u}}}
\label{eqn:proof:lemma:secondmoment:a} \\
    &=
    \E\brackets{ \exp \parenth{-2\sumlminSetinnerprod{E_u}}}
\nonumber
    \\
   &=
   \frac{1}{Z} \int_{\vy \in \cY^p}  
 \exp \parenth{\sumlminSetinnerprod{E}} \exp \parenth{-2\sumlminSetinnerprod{E_u}} 
d\vy \label{eqn:proof:lemma:secondmoment:b} \\
&= \frac{1}{Z} \int_{\vy \in \cY^p}  
 \exp \parenth{\sumlminSetinnerprod{E \setminus E_u}} \exp \parenth{-\sumlminSetinnerprod{E_u}} 
d\vy \nonumber \\
    &= \frac{1}{Z} \int_{\vy_{-u}\in\cY^{p-1}}
    \exp \parenth{\sumlminSetinnerprod{E \setminus E_u}} \int_{y_u \in \cY}
    \exp \parenth{-\sumlminSetinnerprod{E_u}}  dy_u d\vy_{-u} \nonumber \\ 
    &= \frac{1}{Z} \int_{\vy_{-u}\in\cY^{p-1}}
    \exp \parenth{\sumlminSetinnerprod{E \setminus E_u}} \int_{y_u \in \cY + \pi}
    \exp \parenth{-\sumlminSetinnerprod{E_u}}  dy_u d\vy_{-u} \label{eqn:proof:lemma:secondmoment:c} \\ 
    &= \frac{1}{Z} \int_{\vy_{-u}\in\cY^{p-1}}
    \exp \parenth{\sumlminSetinnerprod{E \setminus E_u}} \int_{y_u \in \cY }
    \exp \parenth{\sumlminSetinnerprod{E_u}}  dy_u d\vy_{-u} \label{eqn:proof:lemma:secondmoment:d} \\ 
    &= \frac{1}{Z} \int_{\vy \in\cY^{p}}
    \exp \parenth{\sumlminSetinnerprod{E}} d\vy \label{eqn:proof:lemma:secondmoment:e} \\
    &= 1 \label{eqn:proof:lemma:secondmoment:f} 
\end{align}
where 
\eqref{eqn:proof:lemma:secondmoment:a} follows from
\eqref{eqn:defn:Xij:innerprod}; \eqref{eqn:proof:lemma:secondmoment:b}
follows from
\eqref{eqn:defn:jointdensity:innerprod}; 
\eqref{eqn:proof:lemma:secondmoment:c} follows from the definition of $\tlmstar$ in \eqref{eqn:defn:tlmstar} and the fact that the integral of a periodic function over any period is equivalent; \eqref{eqn:proof:lemma:secondmoment:d} follows from the definition of $\tlmstar$ in \eqref{eqn:defn:tlmstar} and the facts that $\cos(a+\pi)=-\cos(a)$ along with $\sin(a+\pi)=-\sin(a)$; \eqref{eqn:proof:lemma:secondmoment:e} follows from \eqref{eqn:defn:jointdensity:innerprod}; and 
\eqref{eqn:proof:lemma:secondmoment:f} also follows from \eqref{eqn:defn:jointdensity:innerprod}.

The sine component can be shown by analogous argument.
\end{proof}

We again note similarities to \cite[Lemma 2]{vuffray_ising}, where they take advantage of the symmetry of the sum in the exponent with respect a negation of the random variable (e.g., $\sigma_{ij} \xrightarrow[]{} -\sigma_{ij}$). Here, we take advantage of the symmetry of periodic functions in that an integral over any period is equivalent.

\begin{lemma}
\label{lemma:maxX}
For a graphical model defined by \eqref{eq:naturaljointdensity}
with $\beta$ defined by \eqref{eqn:defn:alpha:beta} and maximum degree $d$, 
$\max \braces{\abs{\Xijc(\uthetau^*)},\abs{\Xijs(\uthetau^*)}} \leq \exp(\beta d)$.
\end{lemma}
\begin{proof}
We note that
\begin{align}
    \abs{\Xijc(\uthetau^*)} &= \abs{-\cos(y_j-y_i)\exp\left(-\sum_{(l,m)\in E_u} \thetaijc^*\cos(y_m-y_l)+\thetaijs^*\sin(y_m-y_l)\right)}\nonumber\\
    &\leq \abs{\exp\left(-\sum_{(l,m)\in E_u} \thetaijc^*\cos(y_m-y_l)+\thetaijs^*\sin(y_m-y_l)\right)}\label{eq:follows-bounded-cos}\\
    &= \abs{\exp\left(-\sum_{(l,m)\in E_u} \sqrt{\thetaijc^{*2}+\thetaijs^{*2}}\cos\parenth{y_m-y_l -\arctan\parenth{\frac{\thetaijs}{\thetaijc}}}\right)}\nonumber\\
    &=\abs{\exp\left(-\sum_{(l,m)\in E_u} \kappa_{l,m}^*\cos\parenth{y_m-y_l -\arctan\parenth{\frac{\thetaijs^*}{\thetaijc^*}}}\right)}\nonumber\\
    &\leq \exp(\beta d), \label{eq:follows-beta-def} 
\end{align}
where \eqref{eq:follows-bounded-cos} follows since $\abs{\cos(u)} \leq 1$ and \eqref{eq:follows-beta-def} follows from the definition of $\beta$ given in \eqref{eqn:defn:alpha:beta} and the fact that $|E_u| \leq d$. The sine component holds by analogous argument.
Thus both $\abs{\Xijc}$ and $\abs{\Xijs}$ are upper bounded by $\exp(\beta d)$ and thus so is their maximum.
\end{proof}

With the three components required to invoke Bernstein's inequality assembled, we bound the maximum component of the gradient of the ISO in the following lemma. 
\begin{lemma}
    Consider a graphical model given by \eqref{eq:naturaljointdensity} with p nodes, maximum degree d, and maximal coupling $\beta$ given by \eqref{eqn:defn:alpha:beta}. If $n \geq \exp(2\beta d) \ln(\frac{4p}{\epsilon})$, we have
    \beqa
        \norm{\nabla S_n (\uthetau^*)}_{\infty} \leq 4 \sqrt{\frac{\ln \frac{4p}{\epsilon}}{n}},
    \eeqa
with probability at least $1-\epsilon$. \label{lemma:gradbound}
\end{lemma}

\begin{proof}
We apply Bernstein's inequality to bound the absolute value of one component of the gradient:
\begin{align*}
\mathbb{P}\parenth{\abs{\frac{\partial}{\partial \theta_{ij,c}} S_n(\uthetau^*)} \geq t} \leq 2 \exp\parenth{-\frac{\frac{1}{2}t^2}{\sum_{k=1}^n \E\brackets{\parenth{\frac{\Xkijc}{n}}^2}+ \frac{1}{3}M t }}.
\end{align*}
Note, that the sine and cosine components of the gradient can be bounded in the exact same way. Using Lemmas \ref{lemma:firstmoment}, \ref{lemma:secondmoment}, and \ref{lemma:maxX} we can rewrite this as
\begin{align}
\mathbb{P}\parenth{\abs{\frac{\partial}{\partial \theta_{ij,c}} S_n(\uthetau^*)} \geq t} \leq 2 \exp\parenth{-\frac{\frac{1}{2}t^2n}{1+ \frac{1}{3}\exp(\beta d) t }} . \label{eq:l4-bernstein}
\end{align}

Let $s$ be the argument of the exponent in \eqref{eq:l4-bernstein}. That is, let $s = \frac{\frac{1}{2}t^2n}{1+ \frac{1}{3}\exp(\beta d) t }$. Solving for $t$, we have
\begin{align}
    t = \frac{1}{3}\brackets{\frac{s}{n}\exp(\beta d)+ \sqrt{\parenth{\frac{s}{n}\exp(\beta d)}^2+18\frac{s}{n}}}. \label{eq:l4-t}
\end{align}
We can substitute in \eqref{eq:l4-t} into \eqref{eq:l4-bernstein} and also note that for $n > s \exp(2\beta d)$, we can replace \eqref{eq:l4-t} with $2\sqrt{\frac{s}{n}}$. Thus, we have that
\begin{align*}
    \mathbb{P}\parenth{\abs{\frac{\partial}{\partial \theta_{ij,c}} S_n(\uthetau^*)} \geq 2\sqrt{\frac{s}{n}}} \leq 2 \exp\parenth{-s}.
\end{align*}
Letting $s = \ln\parenth{\frac{4p}{\epsilon}}$, we obtain an expression including $\epsilon$
\begin{align*}
    \mathbb{P}\parenth{\abs{\frac{\partial}{\partial \theta_{ij,c}} S_n(\uthetau^*)} \geq 2\sqrt{\frac{\ln\parenth{\frac{4p}{\epsilon}}}{n}}} \leq \frac{\epsilon}{2p}.
\end{align*}
Applying the union bound over all $2p$ components of the gradient, we can then conclude that
\begin{align*}
    \norm{\nabla S_n(\uthetau^*)}_{\infty} \leq 2\sqrt{\frac{\ln\parenth{\frac{4p}{\epsilon}}}{n}}
\end{align*}
with probability at least $1-\epsilon$ if $n > \ln\parenth{\frac{4p}{\epsilon}} \exp\parenth{2\beta d} $.
\end{proof}

Thus, we have shown Condition 1 can be satisfied with probability at least $1-\epsilon$, if $\lambda > 4 \sqrt{\frac{\ln \frac{4p}{\epsilon}}{n}}$.

\subsection{Restricted Strong Convexity Properties the Wave Model}
To satisfy Condition 2, we need to bound the remainder of the Taylor expansion of the ISO around its minimizer. That is, for some perturbation away from the minimum, we aim for the error to be lower bounded by a constant multiple of the perturbation's norm. Since the ISO is a convex function, intuitively we would like to think about the remainder of the Taylor expansion as some quadratic term and use properties of its Hessian to bound it. In the following lemma, we show that it is indeed bounded by a quadratic.

\begin{lemma}\label{lemma:taylorlb} %a lower bound on the remainder of the taylor expansion 
Define $H_u^n$, the empirical correlation matrix, as 
\begin{equation}
    H_u^n = \frac{1}{n}\sum_{k=1}^n (\utk_u)^T \utk_u.  \label{eq:empirical-cov-mat}
\end{equation} The remainder of the Taylor expansion of the ISO around its minimizer, $\delta S_n(\Delta_u,\utheta^*)$, is lower bounded as
    \[\delta S_n(\Delta_u,\utheta^*) \geq  \frac{\exp(-2\beta d)}{2 + \norm{\Delta}_1}  \Delta_u H_n \Delta_u\].
\end{lemma}
\begin{proof}

Define $\rho(z)$ as
\beqa
\rho(z) &=& \exp(-z)-1+z \label{eqn:defn:rho}
\eeqa
and note that it can be shown that \cite{vuffray_ising}:
\beqa
\rho(z) &\geq& \frac{z^2}{2+\abs{z}} \label{eqn:rho:inequality}
\eeqa

Note that the remainder of the first order Taylor expansion of the ISO be as follows:
\begin{align}
    \delta S_n(\Delta_u,\utheta^*) &= S_n(\utheta^* + \Delta_u) - S_n(\utheta^*) - \innerprod{\nabla S_n(\utheta)}{\Delta_u} \label{eqn:proof:lemma:taylorlb:a}\\
    &= \frac{1}{n} \sum_{k=1}^n \exp\parenth{-\sumijEu \thetasijc \cos(\Yk_j-\Yk_i) - \thetasijs \sin(\Yk_j-\Yk_i)} \nonumber\\
                                & \quad \times \rho\parenth{\sumijEu \Delta_{ij,c} \cos(\Yk_j-\Yk_i) + \Delta_{ij,s} \sin(\Yk_j-\Yk_i)},
\end{align}
where \eqref{eqn:proof:lemma:taylorlb:a}
 follows from \eqref{eqn:defn:Taylorerror}. This can be shown by simply writing out the expression and factoring into this form.
We state two analogous inequalities to what were derived for the Ising ISO \cite[eqn 40-41]{vuffray_ising},
\begin{align}
    \abs{\sum_{(i,j)\in E_u} \Delta_{ij,c} \cos(\Yk_j-\Yk_i) + \Delta_{ij,s} \sin(\Yk_j-\Yk_i) } &\leq \norm{\Delta_u}_1\\
    \exp\left(-\sumijEu \thetasijc \cos(\Yk_j-\Yk_i) - \thetasijs \sin(\Yk_j-\Yk_i)\right) &\geq \exp(-2\beta d).
\end{align}
Combining these two facts, we obtain
\begin{align*}
    \delta S_n(\Delta_u,\utheta^*) &\geq \frac{\exp(-2\beta d)}{2 + \norm{\Delta}_1}\frac{1}{n}\sum_{k=1}^n \left(\sumijEu \Delta_{ij,c} \cos(\Yk_j-\Yk_i) + \Delta_{ij,s} \sin(\Yk_j-\Yk_i)\right)^2\\
    &=  \frac{\exp(-2\beta d)}{2 + \norm{\Delta}_1}\frac{1}{n}\sum_{k=1}^n \left( \Delta_u^T \utk_u \right)^2\\
    &= \frac{\exp(-2\beta d)}{2 + \norm{\Delta}_1}\frac{1}{n}\sum_{k=1}^n  \Delta_u (\utk_u)^T \utk_u \Delta_u\\
    &= \frac{\exp(-2\beta d)}{2 + \norm{\Delta}_1} \Delta_u H_u^n \Delta_u,
\end{align*}
\end{proof}

\begin{lemma}\label{lemma:covconc}
    Let $H_u = \E\brackets{\ut\ut^\top}$ be the correlation matrix. Elements of the empirical correlation matrix given in \eqref{eq:empirical-cov-mat} concentrate such that
    \begin{equation*}
        \abs{H_{u,lm}^n-H_{lm}} \leq \delta,
    \end{equation*}
    with probability at least $1-\epsilon$, when $n \geq \frac{2}{\delta^2}\ln{\frac{4p^2}{\epsilon}}$.
\end{lemma}
\begin{proof}
Note that all entries of $H_u$ are bounded in an interval $[-1,1]$, since they are formed from the products of sines and cosines. Therefore, we can apply Hoeffding's inequality for bounded random variables. Applying the inequality, we have that
\begin{align}
    \mathbb{P}\brackets{\abs{H^n_{u,lm}-H_{u,lm}} > \delta} \leq 2\exp\parenth{-\frac{n\delta^2}{2}}.
\end{align}
Using the union bound over all upper triangular elements of the $2p-2$ sided matrix,
\begin{align}
    \mathbb{P}\brackets{\abs{H^n_{u,lm}-H_{u,lm}} > \delta \hspace{.2cm} \forall (i,j) \in E_u} &\leq \frac{(2p-2)(2p-3)}{2}2\exp\parenth{-\frac{n\delta^2}{2}} \nonumber \\
    &\leq (4p^2 -10p +6) \exp \parenth{-\frac{n\delta^2}{2}}\nonumber\\
    &\leq 4p^2  \exp \parenth{-\frac{n\delta^2}{2}} \label{eq:ecm-concentration},
\end{align}

where \eqref{eq:ecm-concentration} holds for $p\geq0.6$. 
\end{proof}

We will choose $\delta$ appropriately later for the final theorem.\\
\newcommand{\Hfull}{H_{\text{full}}}

\begin{lemma}
    Let $\gamma_u > 0$ be the minimum eigenvalue of the covariance matrix, $H_u $. Then,  $\Delta_u^\top H_u \Delta_u \geq \gamma_u \norm{\Delta_u}_2^2$.
\end{lemma}
\begin{proof}
Consider the full graphical model in \eqref{eq:jointdensity}. Our strategy will use\Cref{lemma:min-exp-fam} and take advantage of the properties of minimal exponential families to show positive definiteness of $H_u$.

By \Cref{lemma:min-exp-fam}, we know that the model given in \eqref{eq:naturaljointdensity} is a minimal exponential family. From the theory of exponential families, we also know that the Hessian of the cumulant function (or log partition function) is strictly positive definite \cite{wainwright2008graphical}. That is,
\begin{align}
    \nabla^2 \log(Z) \succ 0.
\end{align}
However, we also know that the Hessian is the covariance of the sufficient statistic, which means we have the following equivalent statement:
\begin{align}
    \nabla^2 \log(Z) = \E\brackets{\uphi\; \uphi^\top} - \E\brackets{\uphi}\E\brackets{\uphi}^\top \succ 0. \label{eqn:Hessian:cumuluant:posdef}
\end{align}
From \eqref{eqn:Hessian:cumuluant:posdef} we note that $\E\brackets{\uphi\; \uphi^\top}$ is positive definite since $\E\brackets{\uphi}\E\brackets{\uphi}^\top$ is an outer product and thus positive semi-definite. Now, we are in a position to reason about $H_u$.

Let $\Hfull = \E\brackets{\uphi\;\uphi^\top}$. Since $\ut_u$ is the subset of the entries of $\uphi$ that only pertain to entries that involve the node $u$, $H_u$ is a principal minor of $\Hfull$. By Sylvester's criterion, $H_u$ must be positive definite if $\Hfull$ is positive definite. Thus we have that for any $\Delta_u \in \Re^{p-1}$,
\begin{align}
    \Delta_u^\top H_u \Delta_u > \gamma_u \norm{\Delta}_u^2, \label{eqn:proof:lemma:quadraticform}
\end{align}
where $\gamma_u = \lambda_{min}(H_u) > 0$, is the smallest eigenvalue of $H_u$. 
\end{proof}

\begin{lemma}
    Consider the graphical model given by \eqref{eq:jointdensity} with p nodes, maximum degree d, and maximal coupling $\beta$. For all $\epsilon > 0$ and $R > 0$, the ISO satisfies the strong convexity condition with
    \beqa
        \delta S_n(\Delta_u,\utheta) \geq \frac{\gamma_u \exp{\parenth{-2\beta d}}}{4(1+2\sqrt{d}R)}  \norm{\Delta_u}_2^2 ,
    \eeqa
    with probability at least $1-\epsilon$ when $n \geq \frac{8}{\gamma_u^2}\ln \frac{4p^2}{\epsilon}$ such that $\norm{\Delta_u}_2 \leq R$ and $\norm{\Delta_u}_1 \leq 4\sqrt{d}\norm{\Delta_u}_2$. \label{lemma:rsc}
\end{lemma}

\begin{proof}

From \Cref{lemma:taylorlb} we have that
\begin{align}
    \delta S_n(\Delta_u,\utheta^*) &\geq  \frac{\exp(-2\beta d)}{2 + \norm{\Delta}_1}  \Delta_u^\top H^n_u \Delta_u \nonumber\\
    &\geq \frac{\exp(-2\beta d)}{2 + 4\sqrt{d}R}  \Delta_u^\top H^n_u \Delta_u \label{eq:sn-bound-l8}.
\end{align}
We then note that
\beqa
    \Delta_u^\top H^n_u \Delta_u &= \Delta_u^\top H_u \Delta_u + \Delta_u^\top (H^n_u-H_u) \Delta_u \nonumber \\
    &\geq \gamma_u \norm{\Delta_u}_2^2  + \Delta_u^\top (H^n_u-H_u) \Delta_u \label{eq:bound_Hn-H}.
\eeqa 
We use \Cref{lemma:covconc} to bound the second term in \eqref{eq:bound_Hn-H} with $\delta = \frac{\gamma}{2}$

\begin{align}    
    \Delta_u^\top (H^n_u-H_u) \Delta_u &=  \sum_{i,j} \Delta_{u,i}(H_{u,ij}^n - H_{u,ij})\Delta_{u,j} \nonumber\\
    &\geq \sum_{i,j} \Delta_{u,i} (-\delta)\Delta_{u,j} \nonumber \\
    &= -\delta \sum_{i,j} \Delta_{u,i} \Delta_{u,j} \nonumber\\
    &= -\frac{\gamma}{2} \norm{\Delta_u}_2^2 .
\end{align}
Thus we have from \eqref{eqn:proof:lemma:quadraticform} that
\begin{align}
    \Delta_u^\top H_n \Delta_u &\geq \gamma_u \norm{\Delta_u}_2^2 - \frac{\gamma}{2} \norm{\Delta_u}_2^2 \nonumber\\
    &= \frac{\gamma}{2} \norm{\Delta_u}_2^2 \label{eq:dHd-bound}.
\end{align}

Plugging \eqref{eq:dHd-bound} into \eqref{eq:sn-bound-l8} gives us the result with probability at least $1-\epsilon$ if $n \geq \frac{2}{\delta^2}\ln\frac{4p^2}{\epsilon/p^2} = \frac{8}{\gamma^2}\ln\frac{4p^4}{\epsilon} $ by union bounding over every entry of $H^n_u$ with error probability $\frac{\epsilon}{p^2}$ .
\end{proof}

\subsection{Main Theoretical Results}
\begin{theorem}
    Let the penalty parameter $\lambda = 4\sqrt{\frac{\ln (8p/\epsilon_1)}{n}}$. When $n\geq \max\braces{\frac{8}{\gamma^2}\ln \frac{8p^4}{\epsilon_1},\exp\parenth{2\beta d}\ln \parenth{\frac{8p}{\epsilon_1}} }$, we have that 
    \beqa
    \norm{\hat{\utheta}(\lambda)-\utheta^*}_2 \leq \frac{12}{\gamma} \sqrt{\frac{d \ln\parenth{8p/\epsilon_1}}{n}},
    \eeqa
    with probability $1-\epsilon_1$, for  $\norm{\underline{\kappa}}_{\infty} \leq \beta$, maximum degree $d$ of each node, and some constant $\gamma > 0$. \label{thm:squareerror}
\end{theorem}

\begin{proof}
Let $\epsilon$ from \Cref{lemma:gradbound} and \Cref{lemma:rsc} be $\epsilon = \frac{\epsilon_1}{2}$. Using the union bound over both lemmas lets us arrive at the result above with with probability $1-\epsilon_1$. 
\end{proof}

\begin{theorem}
    Let the penalty parameter $\lambda = 4\sqrt{\frac{\ln (8p^2/\epsilon_2)}{n}}$. For minimal coupling $\alpha$, maximal coupling $\beta$, and maximum node degree $d$, we have perfect reconstruction of the edge set $E$ with probability $1-\epsilon_2$ if $n \geq \max\braces{\frac{8}{\gamma^2}\ln \frac{8p^5}{\epsilon_2},\exp\parenth{2\beta d}\ln \parenth{\frac{8p^2}{\epsilon_2}}, 2d \ln(\frac{8p^2}{\epsilon_2})\parenth{\frac{\alpha\gamma}{24 }}^{-2}}$ \label{thm:structure}
\end{theorem}

\begin{proof}
Let $\epsilon_1 = \frac{\epsilon_2}{p}$. We declare an edge to exist if its estimated coupling value $\hat{\kappa}_{ij}$ is greater than or equal to $\frac{\alpha}{2}$. Therefore, the error in \Cref{thm:squareerror} must be bounded by $\frac{\alpha}{2\sqrt{2}}$ ensure no false positives. This is because $\kappa_{ij} = \sqrt{\thetaijc^2+\thetaijs^2}$ and we are controlling the error in the $\theta_{ij}$'s rather than the $\kappa_{ij}$'s. This implies that $n \geq 2d\ln(\frac{8p}{\epsilon_1})\parenth{\frac{\alpha\gamma}{48 }}^{-2}$ for a single node. We then apply the union bound to all nodes with probability of error $\frac{\epsilon_2}{p}$, which gives us the desired result with probability $1-\epsilon_2$. 
\end{proof}

\section{Experiments}\label{sec:experiments}

In this section, we describe computational experiments performed to test both the Chow-Liu and ISO algorithms on data. First, we start by examining the graph recovery properties of the inference procedures when the data are generated according to the model. Then, investigate examples of the algorithm's performance in settings with wave-like data.

\subsection{Model Performance and Computation Speed}
\begin{figure}[htp]
    \centering    \includegraphics[width=\linewidth]{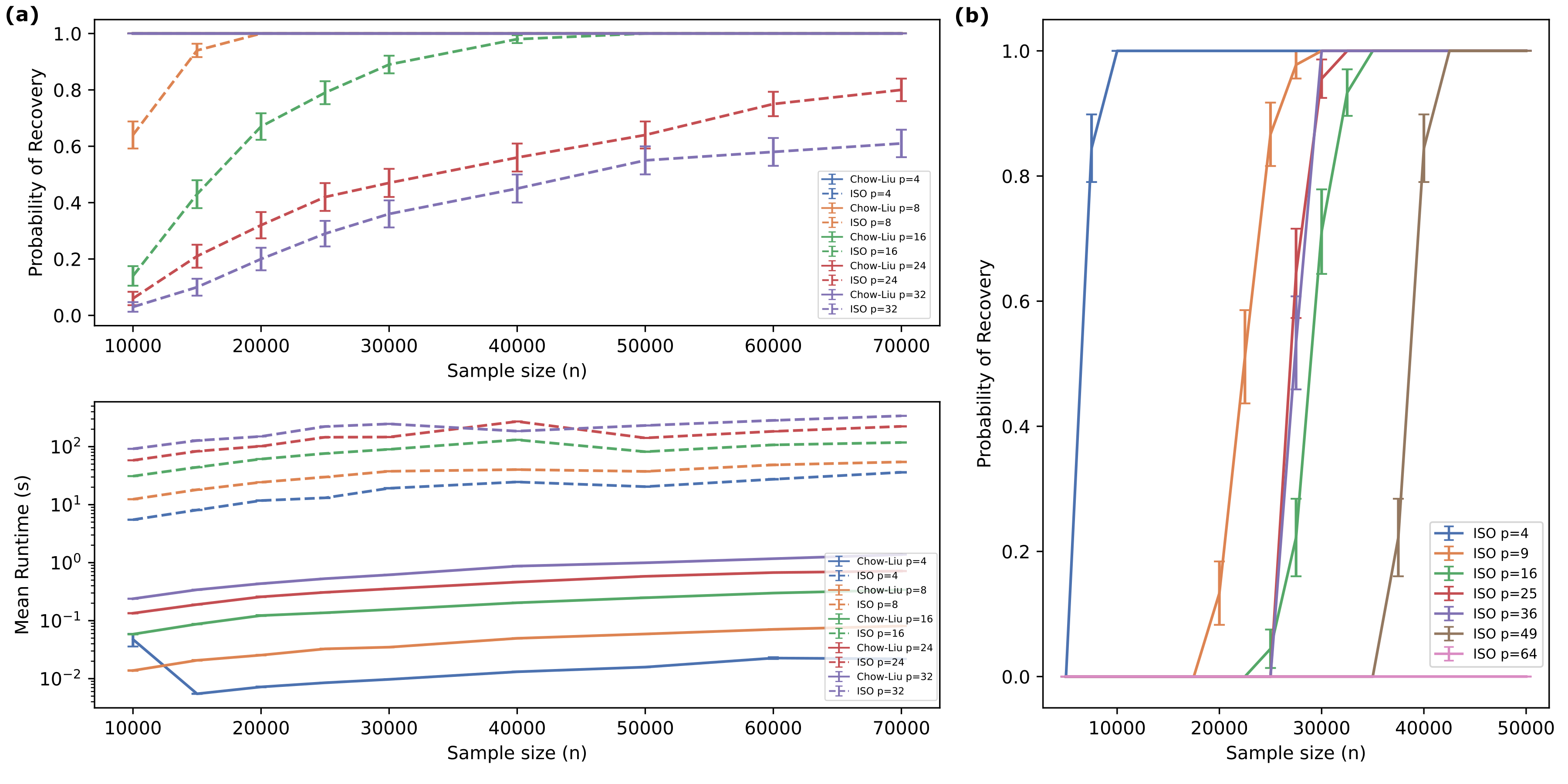}
    \caption{(a) Model Performance Comparison between the Chow-Liu Algorithm and the ISO algorithm. Data generated were according to a tree distribution. (top) Fraction of true graphs recovered over varying sample sizes and numbers of nodes. (bottom) Computation time over varying sample sizes and numbers of nodes. The Chow-Liu gives perfect recovery of the graph structure, while the ISO approaches with increasing sample size.
    (b) Performance of ISO algorithm in graph recovery on four connected graphs.}
    \label{fig:2}
\end{figure}

\subsubsection{Chow-Liu vs ISO on Dependence Trees}
The Chow-Liu dependence tree algorithm is widely used due to its efficiency in terms of computational and space complexity. However, it is limited to modeling dependence trees. On the other hand, the ISO algorithm is more flexible and capable of inferring arbitrary graph structures, but requires more computation. This makes it important to analyze the trade-offs between speed and performance for both methods.

To investigate these trade-offs, we first generate data from \eqref{eq:jointdensity} while constraining the structure to a dependence tree. This ensures that the Chow-Liu algorithm can still recover the graph structure. We vary two parameters: the sample size (n) and the graph size (number of nodes, p). For each (n,p) pair, we generate 100 random dependence trees and draw n samples from each tree. For these trees all $\kappa_{ij}$ are set to 1.

Next, we fit both the Chow-Liu and ISO models to the datasets and calculate the fraction of the 100 randomly generated graphs correctly recovered by each method. This serves as a measure of performance. Additionally, we record the computational time for both algorithms. For the ISO algorithm, we set $\epsilon_2 = 0.05$ and define the regularization parameter $\lambda = 4\sqrt{\frac{\ln (8p^2/\epsilon_2)}{n}}$, to be consistent with \Cref{thm:structure}. The regularization parameter and probability of error are the same as in the last experiment. 

Results of the experiment are in \Cref{fig:2}a.
We see the expected performance in graph recovery of the Chow-Liu and ISO methods, where the error bars are one standard deviation of the mean. Across all tested node numbers and sample sizes, we observe that the Chow-Liu method achieves perfect recovery of the dependence tree every time, while the ISO algorithm achieves modest recovery that increases with sample size and decreases with graph size. We notice that the computational time for the Chow-Liu is approximately 3 orders of magnitude faster than the ISO. Note, that in this experiment none of the computations were parallelized in the ISO. However, if one had $K$ cores, in principle, one could achieve up to a $\min\{K,p\}$ fold increase in computation speed.

\subsubsection{ISO on Four-Connected Graphs}
In order to test the ISO's ability to recover graph structure in general, we also test the algorithm on another set of data from the generative model. In this experiment, we arrange nodes on a square grid of varying sizes. Each node in the grid is connected to its neighbors above, below, to the left, and to the right of itself. If there are no immediate neighbors in any of the directions, then the edge wraps around to the other side. This implies that each node has degree 4. We set $\kappa_{ij}=1$, as in the previous experiment, and generate samples via Gibbs sampling, which is a well know Markov chain Monte Carlo  method  \cite{casella1992explaining}. To perform the Gibbs sampling procedure, we initialize our sample with a uniform random vector on $[-\pi,\pi)^p$. We then re-sample every component of the random vector according to the univariate conditional given in \Cref{subsec:univariate_conditional}, using the most recently sampled components every time the univariate conditional is computed. Performing this process once for each of the p-components generates one new sample. The first 10000 samples were discarded as burn-in samples to allow for the settling time of the Markov chain. Since the ISO is a more computationally expensive method, we run 45 trials for each (n,p) pair and then tabulate the fraction of true graphs recovered at each set of parameters. Note for $p=4$, each node is two-connected since the neighbors to up/down and left/right are the same. 

In \Cref{fig:2}b, we see the analogous recovery graph to \Cref{fig:2}a, but for the ISO over four-connected graphs. We note that the performance is as expected, with larger graphs requiring larger sample sizes to achieve the same probability of recovery.

\subsection{Hypothesis Testing on Simulated Wave Data}

\begin{figure}[htp]
    \centering    \includegraphics[width=\linewidth]{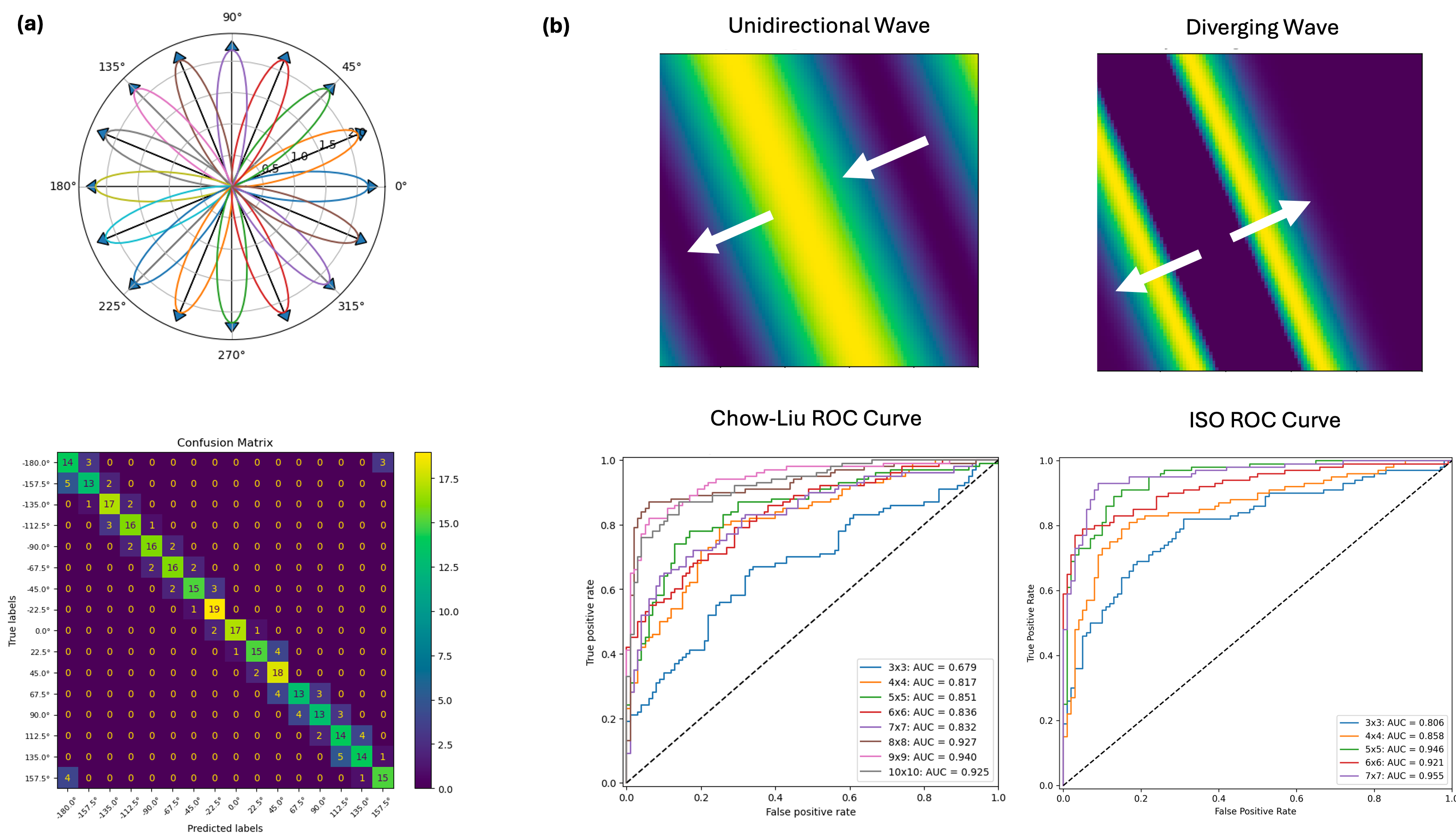}
    \caption{Classification performance on the Chow-Liu and ISO methods. (a) Classification of unidirectional plane waves equally spaced in 16 directions using only the Chow-Liu. The top plot displays the distribution in the direction of each plane wave class. The confusion matrix on the bottom shows good classification with high accuracy. Notice that misclassifications only occur between waves pointing in adjacent directions. (b) Classification between unidirectional and diverging waves using the Chow-Liu method and ISO method. The top plot displays examples of the differing types of waves and their propagation directions. The bottom ROC curve demonstrate diminishing classification returns with higher spatial resolution. Note that the ISO method returns better classification performance (higher AUC) than the Chow-Liu at the same spatial resolutions.}
    \label{fig:3}
\end{figure}

In this section, we describe two experiments done to test the hypothesis testing capability of our wave model on data from the generative model and simulated wave data. In the first experiment, we perform M-ary hypothesis testing on multidirectional plane waves. Here we test only the Chow-Liu Method for hypothesis testing, since M-ary hypothesis testing requires full evaluation of the joint distribution, which is computationally expensive for the full model. In the second experiment, we perform binary hypothesis testing between two classes of waves: unidirectional and diverging. We vary the spatial resolution of our array of data to evaluate classification capability under different measurement conditions. To understand tradeoffs in performance, we evaluate both the Chow-Liu and full wave model using ROC analysis.

\subsubsection{Multidirectional Plane Wave Classification}

To evaluate the model, we use the classical M-ary hypothesis testing framework, where the goal is to classify \( M = 16 \) equally spaced wave propagation directions based on noisy plane wave data.  

We generate the data by first setting up an \( 8 \times 8 \) grid of \((x, y)\) locations for measurements. The phase data is defined as:  
\[
\phi(x, y, t) = K_x x + K_y y - \omega t,
\]
where \( K_x \) and \( K_y \) are the wavenumbers in the \( x \) and \( y \) directions, and \( \omega \) is the angular frequency. Data is sampled at \( T = 200 \) time points with a period of \( \Delta t = 0.02 \) seconds. For each class, the wavenumbers are set as \( K_{x,m} = \cos(\theta_m + \xi) \) and \( K_{y,m} = \sin(\theta_m + \xi) \), where \( \theta_m = \frac{\pi}{8}m \) (for \( m = 0, \dots, 15 \)), and \( \xi \sim VM(0, 30) \) introduces noise via a von-Mises distribution.  

Simulated measurements with noise are given by:  
\[
\underline{v}_m(t) = \cos(\underline{\phi}_m(t)) + \underline{\eta}_t, \quad \text{where } \eta_t \sim \mathcal{N}\left(0, \frac{1}{\sqrt{2}}I\right).
\]  
Phase estimates are extracted using the Hilbert transform by forming the analytic signal (for each component):  
\[
z_{m,i}(t) = v_{m,i}(t) + j H\{v_{m,i}(t)\},
\]  
where $H \braces{\cdot}$ is  the Hilbert transform. We can rewrite the analytic signal as the complex exponential,
\[
z_{m,i}(t) = r_{m,i}(t)e^{j\phi_{m,i}(t)}
\]  
We denote $r_{m,i}(t)$ as the instantaneous amplitude and $\theta_{m,i}(t)$ as the instantaneous phase. The instantaneous phase is used as our estimate of the true phase in the simulation. 

To test the model, we generate 100 datasets for each of the $M$ directions and separate data into train/test sets using 80/20 split stratified by class. We use the maximum likelihood decision rule to classify each test dataset using \eqref{eq:chowliujoint}.

Results for the this experiment are shown in \Cref{fig:3}a. The top plot demonstrates the distribution wave direction for each of the 16 classes of waves and the bottom plot shows a confusion matrix of M-ary hypothesis testing problem. The confusion matrix shows that the Chow-Liu algorithm was able to correctly classify each of the wave directions with high accuracy, achieving 65-95\% accuracy ($76.6\% \pm 8.8\%$) on inferring the exact direction. We note that whenever the classifier made a mistake the mistake was only on the off-diagonal. In other words, only the adjacent two directions clockwise and counter-clockwise to the actual wave direction were predicted when there was a mistake. This demonstrates the ability of our wave model to accurately describe spatial patterns in the phase related to direction of propagation.

\subsubsection{Sensor Array Density Evaluation}
In this experiment, we consider binary hypothesis testing between a normally propagating wave in one direction vs. a diverging wave in two directions, as was previously done using deep learning methods \cite{agrusa2019deep}. Here, we take inspiration from gastric electrophysiology, in which the traveling electrical gastric wave exhibits different characteristics on the stomach in normal individuals as compared to those with delayed gastric emptying. In normal subjects, the gastric wave travels from the top of the stomach (the fundus) to the bottom of the stomach (the antrum) in a regular fashion to push food through the digestive system. However, in those with delayed gastric emptying it has been shown that these slow waves can have abnormal initiations from the antrum and propagate upwards, causing delayed gastric emptying  and/or worsening symptoms \cite{ordog2008interstitial,ANGELI201556,gharibans2019spatial,agrusa2022robust}.

To simulate this class of waves, we consider a simple radial wave on a $1\times1$ unit grid. We let the radial wave have a large semi-major axis and small semi-minor axis compared to the grid size. This is to allow for a diverging wave if the center is on the grid and a unidirectional wave if the center is off-grid. More concretely, we generate the phase data according to
\begin{align*}
    \phi(x,y,t) = K_r\sqrt{0.09\brackets{{(x-c_x)\cos(\theta) - (y-c_y)\sin(\theta)}}^2 + 225\brackets{(x-c_x)\sin(\theta) + (y-c_y)\cos(\theta)}^2 } - \omega t,
\end{align*}

\newcommand{\uc}{\underline{c}}
where $K_r$ is the nominal radial wavenumber, $\theta$ is the angle of rotation, $\omega$ is the angular frequency, and $\uc = \begin{bmatrix}
    c_x & c_y
\end{bmatrix}$ is the center of the wave. To allow for variability between groups and wave direction (to simulate inter-subject variability in anatomy), we allow the parameters of the model to have the following distributions
\begin{align*}
    \begin{matrix}
        K_r \sim \mathcal{N}(1,0.1) & \theta \sim VM(\frac{\pi}{4},60)\\
        \uc_0 \sim \mathcal{U}\parenth{[1,1.5]\times [1,1.5]} & \uc_1 \sim \mathcal{U}\parenth{[0,.75]\times [0,.75]},
    \end{matrix}
\end{align*}
where $\uc_0$ and $\uc_1$ indicate the centers of the wave for the unidirectional and diverging waves, respectively. Note the distributions of $K_r$ and $\theta$ are the same among classes, and $\omega = \frac{2\pi}{20}$ is held constant. The wave data is then computed as
\begin{align*}
    v(x,y,t) = \cos\parenth{\phi(x,y,t)} + \xi_t,
\end{align*}
where $\xi_t \sim \mathcal{N}\parenth{0,\frac{1}{\sqrt{2}}}$ is measurement noise to simulate 0dB signal-to-noise ratio environments.

For statistical analysis, the simulated data is preprocessed through a standard signal processing pipeline to filter out the wave of interest from the measurements and extract phase. To do this we use an $8^{th}$ order zero-phase butterworth filter with a passband of $[0.03, 0.07] $ Hz, since an angular frequency of $\frac{2 \pi}{20}$ is equivalent to $0.05 $ Hz. We then use the Hilbert transform as described in the previous section to obtain estimates of $\underline{\phi}(t)$. To get a sense of the classification ability of this model, 500 datasets of each class, with $T = 200$ time points spaced apart by $\Delta t = 0.02$s are generated for train/test purposes. Note that all 1000 datasets differ by all of $K_r$, $\theta$, and $\underline{c}_0$ or $\underline{c}_1$. We using an 80/20 stratified split of the datasets into train/test sets. The Chow-Liu and ISO method are then applied to the data with $\lambda = 4\sqrt{\frac{\ln (8p^2/\epsilon_2)}{n}}$ and $\epsilon_2 = 0.05$. We note here, that during the training process, each of the time points are treated as i.i.d. To obtain good parameter estimates for the model, we perform the two-step procedure where we 1) fit the regularized model to identify the edge set and then 2) fit the unregularized model with only the edges identified by the regularized model (i.e., hard code all edges not identified in 1 to be zero). We use the log likelihood ratio (LLR) test to perform the binary classification. ROC curves are shown over a range of sensor densities from 3x3 to 10x10 for the Chow-Liu and 3x3 to 7x7 for the ISO in \Cref{fig:3}.

It is important to note, that for binary classification using LLR test we do not actually need to evaluate the partition function in order to get ROC curve analysis. This can be illustrated by first writing down the LLR test,
\begin{align*}
    LL(\underline{y};\underline{\kappa}_1,\underline{\mu}_1) - LL(\underline{y};\underline{\kappa}_0,\underline{\mu}_0)\underset{\hat{H}= 0}{\overset{\hat{H}=1}{\gtrless}} \tau,
\end{align*}
where $\underline{\kappa}_1,\underline{\mu}_1$  and $\underline{\kappa}_0,\underline{\mu}_0$ are the parameter sets for class 1 and 0, $LL$ is the log-likelihood, and $\tau$ is the decision threshold. Since the partition function is only a function of the parameters, and not the data, we can rewrite this as
\begin{align*}
    \frac{1}{n}\sum_{k=1}^n\brackets{\sum_{(i,j)\in E} \kappa_{1,ij}\cos(\yk_j-\yk_i-\mu_{1,ij}) - \sum_{(i,j)\in E} \kappa_{0,ij}\cos(\yk_j-\yk_i-\mu_{0,ij})}&\underset{\hat{H}= 0}{\overset{\hat{H}=1}{\gtrless}} \tau - Z(\underline{\kappa}_1,\underline{\mu}_1) + Z(\underline{\kappa}_0,\underline{\mu}_0)\\
    &\underset{\hat{H}= 0}{\overset{\hat{H}=1}{\gtrless}} \alpha,
\end{align*}
which is simply a shifted threshold. Since the ROC curve is parameterized by the choice of threshold, this change only affects the parametrization, and not the ROC curve itself.

Results for the unidirectional/diverging classification problem are shown in \Cref{fig:3}b. In this problem, we use both the Chow-Liu method and the full wave model fit using the ISO in order to run ROC analysis. The top plots in \Cref{fig:3}b represent a snapshot of a unidirectional and a diverging wave to give a sense of the data being classified. The bottom two plots show the ROC analysis for the Chow-Liu approximation and the ISO derived full wave model. Focusing on the Chow-Liu plot first, we notice that, as expected, increasing sensor density increases classification performance, with a roughly increasing AUC with sensor density. However, we also note that there is a point of diminishing returns in that the performance between the $8 \times 8$, $9 \times 9$, and $10 \times 10$ arrays are all very similar and seem to asymptotically be reaching their limit of classification.  This has important implications for the design of emerging high-density measurement systems tailored to tracking wave patterns \cite{liu2021multimodal,ramezani2024high,kurniawan2022electrochemical}.  In the bottom right plot, we see the analogous curve for the ISO. Immediately, we can see that compared to the Chow-Liu algorithm, the full wave model performs better for the same sensor density. This implies that the full wave model may be necessary for wave-like phenomena, and that assuming a dependence tree structure will result in loss of performance and may necessitate more sensors to enable accurate classification.

\section{Discussion}\label{sec:discussion}

In this work, we present an exponential family probabilistic graphical model for modeling multivariate phase relationships and subsequently present methods to either fit a constrained approximation, using the Chow-Liu dependence tree, or estimate the full model in generality, using the interaction screening objective estimator. We then show analysis for the ISO estimator, analogous to the Ising model, to demonstrate information-theoretically good sample efficiency in estimating parameters of the model \cite{vuffray_ising}. We recognize that while the distribution is the same as that in \cite{cadieu2010phase}, our work for the first time provides  theoretical guarantees for both the structure learning and sample complexity of parameter estimation. We test this claim in experiments, by generating data from both a dependence tree, to match Chow-Liu assumptions, and a four-connected graph, to test the goodness of the ISO estimator. Here, we demonstrated the tradeoffs between the Chow-Liu and ISO, by showcasing that the Chow-Liu is three orders of magnitude faster and gives perfect recovery of dependence trees, but cannot recover the structure of any graph outside of the class of trees.

In addition, we showcase preliminary results on synthetic wave data in multiple scenarios of hypothesis testing using our model. In the binary hypothesis testing problem, we demonstrate further the tradeoffs between Chow-Liu and the ISO by noting that the ROC performance of the ISO estimator outclassed that of the Chow-Liu, albeit being slower (as shown in the previous experiment). This may be because using the full wave model allows for more flexibility in the phase differences between spatial locations, by allowing for disagreement between parameters, while the Chow-Liu approach does not. For example, in the case of Chow-Liu, if we had three nodes 1-2-3 and $\tilde{\mu}_{12} = \frac{\pi}{4}$ and $\tilde{\mu}_{23} = \frac{\pi}{2}$, the maximum likelihood phase difference between nodes $1$ and $3$ must be $\frac{3\pi}{4}$. However, in the full graphical model, we can have, for example, $\mu_{12} = \frac{\pi}{4}$ and $\mu_{23} = \frac{\pi}{2}$ and $\mu_{13} = \frac{2\pi}{3} \neq\mu_{12} + \mu_{23}$. We must take care to note that while we are confident that this performance generalizes to M-ary hypothesis testing,  it is still a challenging problem to evaluate the likelihood of the full wave model in arbitrary dimensions since the partition function is a generally hard integral to compute. One method of computing this integral is by Monte-Carlo integration, which with enough samples is reliable, but computationally expensive.

One point of note is that our model treats all of our samples across time as independent and identically distributed.  On the one hand, waves are time-varying phenomena where there is correlation of the phases of any oscillator at consecutive time points.  On the flipside, however, structured waves (plane, rotating, etc) as considered in \cref{sec:experiments} have phase {\it differences} between pairs of nodes that are time-invariant.  Our ability to rigorously characterize multivariate phase relationships enables opportunities to test hypotheses of how sensory processes with invariances can be represented in neural systems as spatiotemporal relationships with conserved quantities (with traveling waves as a key example) \cite{keller2024spacetime}.  Future iterations of our model nonetheless can  include a time-dependent component to capture higher order interaction over time between the same and neighbor nodes, with one natural approach being a Markov model on the latent parameters.  

Going forward, we plan to investigate the performance and application of our model and inference methods in real wave-like datasets of interest, including datasets with traveling waves across the cortex of the brain and the traveling electrical waves of the gastrointestinal system. In pursuit of the real-world applicability of our model, we also plan to develop  methods to incorporate complex higher-order and time-dependent interactions in the data so that we can fully characterize the intricacies of these phenomena.

More generally, we aim to understand how we can characterize functional relationships of oscillators within the brain and body.  Since $\kappa$ entries in our model being zero encodes conditional independence, our model has an ability to obviate indirect effects in a manner analogous to the success of the graphical lasso for conventional measures of brain connectivity with multivariate Gaussian assumptions \cite{smith2011network,huang2010learning}.
Tools to compare functional brain networks whose functional connections were each estimated with $\ell_1$ regularization (as is done here) can leverage recent approaches developed for graphical lasso that enable obtaining confidence intervals on edge differences in the estimated distributions across  subjects 
\cite{belilovsky2016testing}.  Further building upon this analogy, natural extensions of our $\ell_1$ regularized coupled oscillation estimation approach include exploring the connection of our ISO problem formulation with  distributionally robust estimation, as has been done with covariance estimation and graphical lasso \cite{blanchet2019optimal,cisneros2020distributionally,blanchet2021statistical}.

\section*{Acknowledgements}

We would like to acknowledge support from the Stanford Data Science Scholars Program, the Siebel Scholars Program, the Wu Tsai Neurosciences Institute, and NIH 5U01NS131914.

\bibliographystyle{ieeetr}
\bibliography{template}

\begin{thebibliography}{10}

\bibitem{surface_em_waves}
J.~Polo~Jr. and A.~Lakhtakia, ``Surface electromagnetic waves: A review,'' {\em
  Laser \& Photonics Reviews}, vol.~5, no.~2, pp.~234--246, 2011.

\bibitem{plant_sound}
M.~Pagano and S.~Del~Prete, ``Symphonies of growth: Unveiling the impact of
  sound waves on plant physiology and productivity,'' {\em Biology}, vol.~13,
  no.~5326, 2024.
\newblock Citation Key: biology13050326 tex.pubmedid: 38785808.

\bibitem{cardiac_Waves}
B.~E. Steinberg, L.~Glass, A.~Shrier, and G.~Bub, ``The role of heterogeneities
  and intercellular coupling in wave propagation in cardiac tissue,'' {\em
  Philosophical Transactions of the Royal Society A: Mathematical, Physical and
  Engineering Sciences}, vol.~364, no.~1842, pp.~1299--1311, 2006.

\bibitem{Muller_Chavane_Reynolds_Sejnowski_2018}
L.~Muller, F.~Chavane, J.~Reynolds, and T.~J. Sejnowski, ``Cortical travelling
  waves: mechanisms and computational principles,'' {\em Nature Reviews
  Neuroscience}, vol.~19, p.~255–268, May 2018.

\bibitem{muller2016rotating}
L.~Muller, G.~Piantoni, D.~Koller, S.~S. Cash, E.~Halgren, and T.~J. Sejnowski,
  ``Rotating waves during human sleep spindles organize global patterns of
  activity that repeat precisely through the night,'' {\em Elife}, vol.~5,
  p.~e17267, 2016.

\bibitem{davis2024horizontal}
Z.~W. Davis, A.~Busch, C.~Steward, L.~Muller, and J.~Reynolds, ``Horizontal
  cortical connections shape intrinsic traveling waves into feature-selective
  motifs that regulate perceptual sensitivity,'' {\em Cell reports}, vol.~43,
  no.~9, 2024.

\bibitem{vidaurre2018spontaneous}
D.~Vidaurre, L.~T. Hunt, A.~J. Quinn, B.~A. Hunt, M.~J. Brookes, A.~C. Nobre,
  and M.~W. Woolrich, ``Spontaneous cortical activity transiently organises
  into frequency specific phase-coupling networks,'' {\em Nature
  communications}, vol.~9, no.~1, p.~2987, 2018.

\bibitem{rubino2006propagating}
D.~Rubino, K.~A. Robbins, and N.~G. Hatsopoulos, ``Propagating waves mediate
  information transfer in the motor cortex,'' {\em Nature neuroscience},
  vol.~9, no.~12, pp.~1549--1557, 2006.

\bibitem{davis2020spontaneous}
Z.~W. Davis, L.~Muller, J.~Martinez-Trujillo, T.~Sejnowski, and J.~H. Reynolds,
  ``Spontaneous travelling cortical waves gate perception in behaving
  primates,'' {\em Nature}, vol.~587, no.~7834, pp.~432--436, 2020.

\bibitem{gonzales2025touch}
D.~L. Gonzales, H.~F. Khan, H.~V. Keri, S.~Yadav, C.~Steward, L.~E. Muller,
  S.~R. Pluta, and K.~Jayant, ``Touch-evoked traveling waves establish a
  translaminar spacetime code,'' {\em Science Advances}, vol.~11, no.~5,
  p.~eadr4038, 2025.

\bibitem{canolty2010oscillatory}
R.~T. Canolty, K.~Ganguly, S.~W. Kennerley, C.~F. Cadieu, K.~Koepsell, J.~D.
  Wallis, and J.~M. Carmena, ``Oscillatory phase coupling coordinates
  anatomically dispersed functional cell assemblies,'' {\em Proceedings of the
  National Academy of Sciences}, vol.~107, no.~40, pp.~17356--17361, 2010.

\bibitem{gharibans2019spatial}
A.~A. Gharibans, T.~P. Coleman, H.~Mousa, and D.~C. Kunkel, ``Spatial patterns
  from high-resolution electrogastrography correlate with severity of symptoms
  in patients with functional dyspepsia and gastroparesis,'' {\em Clinical
  Gastroenterology and Hepatology}, vol.~17, no.~13, pp.~2668--2677, 2019.

\bibitem{agrusa2022robust}
A.~S. Agrusa, D.~C. Kunkel, and T.~P. Coleman, ``Robust regression and optimal
  transport methods to predict gastrointestinal disease etiology from high
  resolution egg and symptom severity,'' {\em IEEE Transactions on Biomedical
  Engineering}, vol.~69, no.~11, pp.~3313--3325, 2022.

\bibitem{reaction_diffusion}
H.~Berestycki and F.~Hamel, ``Generalized travelling waves for reaction-diusion
  equations,'' {\em Contemp. Math.}, vol.~446, 01 2007.

\bibitem{protein_pred}
S.~Kim, A.~Sengupta, and B.~Arnold, ``A multivariate circular distribution with
  applications to the protein structure prediction problem,'' {\em Journal of
  Multivariate Analysis}, vol.~143, 10 2015.

\bibitem{wind_direction}
X.~Qin, J.~Zhang, and X.~Yan, ``A new circular distribution and its application
  to wind data,'' {\em Journal of Mathematics Research}, vol.~2, pp.~12--17, 08
  2010.

\bibitem{MvM}
K.~Mardia, G.~Hughes, C.~Taylor, and H.~Singh, ``A multivariate {von Mises}
  distribution with applications to bioinformatics,'' {\em Canadian Journal of
  Statistics}, vol.~36, pp.~99 -- 109, 03 2008.

\bibitem{mGvM}
A.~Navarro, J.~Frellsen, and R.~Turner, ``The multivariate generalised {von
  Mises}: Inference and applications,'' {\em Proceedings of the AAAI Conference
  on Artificial Intelligence}, vol.~31, 02 2016.

\bibitem{wrapped_dist_review}
W.~Bell and S.~Nadarajah, ``A review of wrapped distributions for circular
  data,'' {\em Mathematics}, vol.~12, p.~2440, 08 2024.

\bibitem{huang2010learning}
S.~Huang, J.~Li, L.~Sun, J.~Ye, A.~Fleisher, T.~Wu, K.~Chen, E.~Reiman,
  A.~D.~N. Initiative, {\em et~al.}, ``Learning brain connectivity of
  alzheimer's disease by sparse inverse covariance estimation,'' {\em
  NeuroImage}, vol.~50, no.~3, pp.~935--949, 2010.

\bibitem{smith2011network}
S.~M. Smith, K.~L. Miller, G.~Salimi-Khorshidi, M.~Webster, C.~F. Beckmann,
  T.~E. Nichols, J.~D. Ramsey, and M.~W. Woolrich, ``Network modelling methods
  for fmri,'' {\em Neuroimage}, vol.~54, no.~2, pp.~875--891, 2011.

\bibitem{friedman2008sparse}
J.~Friedman, T.~Hastie, and R.~Tibshirani, ``Sparse inverse covariance
  estimation with the graphical lasso,'' {\em Biostatistics}, vol.~9, no.~3,
  pp.~432--441, 2008.

\bibitem{vuffray_ising}
M.~Vuffray, S.~Misra, A.~Lokhov, and M.~Chertkov, ``Interaction screening:
  Efficient and sample-optimal learning of ising models,'' {\em Advances in
  neural information processing systems}, vol.~29, 2016.

\bibitem{Shah_Shah_Wornell_2021}
A.~Shah, D.~Shah, and G.~Wornell, ``On learning continuous pairwise {Markov}
  random fields,'' in {\em Proceedings of The 24th International Conference on
  Artificial Intelligence and Statistics}, p.~1153–1161, PMLR, Mar. 2021.

\bibitem{cadieu2010phase}
C.~F. Cadieu and K.~Koepsell, ``Phase coupling estimation from multivariate
  phase statistics,'' {\em Neural computation}, vol.~22, no.~12,
  pp.~3107--3126, 2010.

\bibitem{Chow_Liu_1968}
C.~Chow and C.~Liu, ``Approximating discrete probability distributions with
  dependence trees,'' {\em IEEE Transactions on Information Theory}, vol.~14,
  p.~462–467, May 1968.

\bibitem{cover1999elements}
T.~M. Cover, {\em Elements of information theory}.
\newblock John Wiley \& Sons, 1999.

\bibitem{wainwright2008graphical}
M.~J. Wainwright, M.~I. Jordan, {\em et~al.}, ``Graphical models, exponential
  families, and variational inference,'' {\em Foundations and
  Trends{\textregistered} in Machine Learning}, vol.~1, no.~1--2, pp.~1--305,
  2008.

\bibitem{muller2021algebraic}
L.~Muller, J.~Min{\'a}{\v{c}}, and T.~T. Nguyen, ``Algebraic approach to the
  kuramoto model,'' {\em Physical Review E}, vol.~104, no.~2, p.~L022201, 2021.

\bibitem{davis2021spontaneous}
Z.~W. Davis, G.~B. Benigno, C.~Fletterman, T.~Desbordes, C.~Steward, T.~J.
  Sejnowski, J.~H.~Reynolds, and L.~Muller, ``Spontaneous traveling waves
  naturally emerge from horizontal fiber time delays and travel through locally
  asynchronous-irregular states,'' {\em Nature Communications}, vol.~12, no.~1,
  p.~6057, 2021.

\bibitem{budzinski2023analytical}
R.~C. Budzinski, T.~T. Nguyen, G.~B. Benigno, J.~{\DJ}o{\`a}n,
  J.~Min{\'a}{\v{c}}, T.~J. Sejnowski, and L.~E. Muller, ``Analytical
  prediction of specific spatiotemporal patterns in nonlinear oscillator
  networks with distance-dependent time delays,'' {\em Physical Review
  Research}, vol.~5, no.~1, p.~013159, 2023.

\bibitem{romano2002xy}
S.~Romano and V.~Zagrebnov, ``On the {XY} model and its generalizations,'' {\em
  Physics Letters A}, vol.~301, no.~5-6, pp.~402--407, 2002.

\bibitem{Negahban_Ravikumar_Wainwright_Yu_2012}
S.~N. Negahban, P.~Ravikumar, M.~J. Wainwright, and B.~Yu, ``A unified
  framework for high-dimensional analysis of $m$-estimators with decomposable
  regularizers,'' {\em Statistical Science}, vol.~27, p.~538–557, Nov. 2012.

\bibitem{casella1992explaining}
G.~Casella and E.~I. George, ``Explaining the {Gibbs} sampler,'' {\em The
  American Statistician}, vol.~46, no.~3, pp.~167--174, 1992.

\bibitem{agrusa2019deep}
A.~S. Agrusa, A.~A. Gharibans, A.~A. Allegra, D.~C. Kunkel, and T.~P. Coleman,
  ``A deep convolutional neural network approach to classify normal and
  abnormal gastric slow wave initiation from the high resolution
  electrogastrogram,'' {\em IEEE transactions on biomedical engineering},
  vol.~67, no.~3, pp.~854--867, 2019.

\bibitem{ordog2008interstitial}
T.~{\"O}rd{\"o}g, ``Interstitial cells of {Cajal} in diabetic
  gastroenteropathy,'' {\em Neurogastroenterology \& Motility}, vol.~20, no.~1,
  pp.~8--18, 2008.

\bibitem{ANGELI201556}
T.~R. Angeli, L.~K. Cheng, P.~Du, T.~H.-H. Wang, C.~E. Bernard, M.-G.
  Vannucchi, M.~S. Faussone-Pellegrini, C.~Lahr, R.~Vather, J.~A. Windsor,
  G.~Farrugia, T.~L. Abell, and G.~O’Grady, ``Loss of interstitial cells of
  cajal and patterns of gastric dysrhythmia in patients with chronic
  unexplained nausea and vomiting,'' {\em Gastroenterology}, vol.~149, no.~1,
  pp.~56--66.e5, 2015.

\bibitem{liu2021multimodal}
X.~Liu, C.~Ren, Y.~Lu, Y.~Liu, J.-H. Kim, S.~Leutgeb, T.~Komiyama, and
  D.~Kuzum, ``Multimodal neural recordings with neuro-fitm uncover diverse
  patterns of cortical--hippocampal interactions,'' {\em Nature neuroscience},
  vol.~24, no.~6, pp.~886--896, 2021.

\bibitem{ramezani2024high}
M.~Ramezani, J.-H. Kim, X.~Liu, C.~Ren, A.~Alothman, C.~De-Eknamkul, M.~N.
  Wilson, E.~Cubukcu, V.~Gilja, T.~Komiyama, {\em et~al.}, ``High-density
  transparent graphene arrays for predicting cellular calcium activity at depth
  from surface potential recordings,'' {\em Nature nanotechnology}, vol.~19,
  no.~4, pp.~504--513, 2024.

\bibitem{kurniawan2022electrochemical}
J.~F. Kurniawan, A.~B. Allegra, T.~Pham, A.~K. Nguyen, N.~L. Sit, B.~Tjhia,
  A.~J. Shin, and T.~P. Coleman, ``Electrochemical performance study of
  {Ag/AgCl} and {Au} flexible electrodes for unobtrusive monitoring of human
  biopotentials,'' {\em Nano Select}, vol.~3, no.~8, pp.~1277--1287, 2022.

\bibitem{keller2024spacetime}
T.~A. Keller, L.~Muller, T.~J. Sejnowski, and M.~Welling, ``A spacetime
  perspective on dynamical computation in neural information processing
  systems,'' {\em arXiv preprint arXiv:2409.13669}, 2024.

\bibitem{belilovsky2016testing}
E.~Belilovsky, G.~Varoquaux, and M.~B. Blaschko, ``Testing for differences in
  {Gaussian} graphical models: Applications to brain connectivity,'' {\em
  Advances in neural information processing systems}, vol.~29, 2016.

\bibitem{blanchet2019optimal}
J.~Blanchet and N.~Si, ``Optimal uncertainty size in distributionally robust
  inverse covariance estimation,'' {\em Operations Research Letters}, vol.~47,
  no.~6, pp.~618--621, 2019.

\bibitem{cisneros2020distributionally}
P.~Cisneros-Velarde, A.~Petersen, and S.-Y. Oh, ``Distributionally robust
  formulation and model selection for the graphical lasso,'' in {\em
  International Conference on Artificial Intelligence and Statistics},
  pp.~756--765, PMLR, 2020.

\bibitem{blanchet2021statistical}
J.~Blanchet, K.~Murthy, and V.~A. Nguyen, ``Statistical analysis of wasserstein
  distributionally robust estimators,'' in {\em Tutorials in Operations
  Research: Emerging optimization methods and modeling techniques with
  applications}, pp.~227--254, INFORMS, 2021.

\end{thebibliography}

\end{document}